\documentclass[aos]{imsart}

\usepackage{color}
\usepackage{verbatim}
\usepackage{float}

\usepackage{amsmath}
\usepackage{amssymb}
\usepackage{graphicx}

\usepackage{times}
\usepackage{bm}
\usepackage{natbib}
\usepackage[plain,noend]{algorithm2e}
\usepackage{amsthm}

\makeatletter

\floatstyle{ruled}
\newfloat{algorithm}{tbp}{loa}
\providecommand{\algorithmname}{Algorithm}
\floatname{algorithm}{\protect\algorithmname}

\numberwithin{equation}{section}

\usepackage{algorithmic}

\newcommand{\diff}{\,\mathrm{d}}

\newcommand{\ascv}{\,\mathrm{a.s.}\,}

\newcommand{\prob}{\mathrm{pr}}

 \newcommand\independent{\protect\mathpalette{\protect\independenT}{\perp}}
\def\independenT#1#2{\mathrel{\rlap{$#1#2$}\mkern2mu{#1#2}}}

\bibpunct{(}{)}{;}{a}{,}{;}

\makeatother

\newtheorem{theorem}{Theorem}
\newtheorem{lemma}{Lemma}

\theoremstyle{definition}
\newtheorem{definition}{Defintition}
\newtheorem{example}{Example}
\newtheorem{remark}{Remark}
\newtheorem{assumption}{Assumption}

\startlocaldefs
\endlocaldefs

\begin{document}

\begin{frontmatter}

\title{Approximate Bayesian Computation via Sufficient Dimension Reduction}
\runtitle{ABC via SDR}
\begin{aug}

\author{\fnms{Xiaolong} \snm{Zhong}\ead[label=e1]{zhongxl@stat.ufl.edu}}

\and
\author{\fnms{Malay} \snm{Ghosh} \thanksref{t1} \ead[label=e2]{ghoshm@stat.ufl.edu}}
\affiliation{University of Florida}

\thankstext{t1}{Supported by an NSF grant}

\address{Department of Statistics \\ 
University of Florida\\
32611, Florida, U.S.A\\
\printead{e1} \\
\phantom{E-mail:\ }  \printead*{e2}}

\runauthor{X. Zhong and M. Ghosh}
\end{aug}

\begin{abstract}
Approximate Bayesian computation (ABC) has gained popularity in recent
years owing to its easy implementation, nice interpretation and good
performance. Its advantages are more visible when one encounters complex
models where maximum likelihood estimation as well as Bayesian analysis
via  Markov chain Monte Carlo demand prohibitively large amount of
time. This paper examines properties of ABC both from a theoretical
as well as from a computational point of view. We consolidate the ABC
theory by proving theorems related to its limiting behaviour. In particular,
we consider partial posteriors, which serve as the first step towards
approximating the full posteriors. Also, a new semi-automatic algorithm
of ABC is proposed using sufficient dimension reduction (SDR) method.
SDR has primarily surfaced in the frequentist literature. But we have
demonstrated in this paper that it has connections with ABC as well.
\end{abstract}



\end{frontmatter}

\section{Introduction}

There are two main objectives of this article. First, we want to provide
some theoretical results related to the currently emerging topic of
approximate Bayesian computation (ABC). The second is to show some
connectivity between ABC and another important emerging topic of research,
namely, sufficient dimension reduction (SDR). While the latter has
surfaced primarily in the\textcolor{black}{{} frequentist's}{{}
}domain of research, it is possible to tie it with ABC as well. In
particular, we want to show how ABC can be carried through nonlinear
SDR. 

Modern science invokes more and more Byzantine stochastic models,
such as stochastic kinetic network (\cite{wilkinson2011stochastic}),
differential equation system (\cite{picchini2014inference}) and
multi-hierarchical model (\cite{jasra2012filtering}), whose computational
complexity and intractability challenge the application of classical
statistical inference. Traditional maximum likelihood methods will
malfunction when the evaluation of likelihoods becomes slow and inaccurate.
Lack of analytical form of the likelihood also \textcolor{black}{undermines
} the usage of Bayesian inferential tools, such as Markov chain Monte
Carlo (MCMC), Laplace approximation (\cite{tierney1986accurate}),
variational Bayes (\cite{jaakkola2000bayesian}) and posterior expansion
(\cite{johnson1970asymptotic}, Zhong and Ghosh%
). 

The ABC methodology \textcolor{black}{stems from} the observation
that the interpretability of the candidate model usually leads to
an applicable sampler of data given parameters, and ingeniously circumvents
the evaluation of likelihood functions. The idea behind ABC can be
summarized as follows: 
\begin{algorithm}[H]
\begin{enumerate}
\item Sample parameters $\theta_{i}$ from the prior distribution $\pi\left(\theta\right)$;
\item Sample data $Z_{i}$ based on the model $f\left(z\mid\theta_{i}\right)$;
\item Compare the simulated data $Z_{i}$ and the observed data $X_{i,\mathrm{obs}}$,
to accept or reject $\theta_{i}$.
\end{enumerate}
\protect\caption{Idea of ABC}

\end{algorithm}
 \cite{rubin1984bayesianly} first mentioned this idea and \cite{tavare1997inferring}
proposed the first version of ABC, while studying population genetics.
The prototype of ABC in recent research was given in \cite{pritchard1999population},
where the comparison of two data sets was simplified to a comparison
of summary statistics $S$ and the accept-reject decision was made
up to a certain error tolerance. 
\begin{algorithm}[h]
\begin{enumerate}
\item Sample parameters $\theta_{i}$ from the prior distribution $\pi\left(\theta\right)$;
\item Sample data $Z_{i}$ based on the model $f\left(z\mid\theta_{i}\right)$;
\item Accept $\theta_{i}$ if $\rho\left(S\left(Z_{i}\right),S\left(X_{\mathrm{obs}}\right)\right)\le\varepsilon$,{{}
}\textcolor{black}{for some metric $\rho$}{.}
\end{enumerate}
\protect\caption{\label{alg:Prichard-ABC}Prichard's Modified ABC}

\end{algorithm}
We can view this algorithm as a modified version of accept-reject
algorithm (\cite{robert2013monte}). The posterior is sampled by
altering the frequency of the proposal distribution, that is, the
prior. Now the full posterior distribution is approximated by the
following two steps (\cite{fearnhead2012constructing}): 
\begin{equation}
\pi\left(\theta\mid X_{\mathrm{obs}}\right)\approx\pi\left(\theta\mid S_{\mathrm{obs}}\right)\approx\pi\left(\theta\mid S_{\mathrm{sim}}\in O\left(S_{\mathrm{obs}},\varepsilon\right)\right),\label{eq:two-step-approx-abc}
\end{equation}
where $O\left(S_{\mathrm{obs}},\varepsilon\right)$ means a neighborhood
defined by the comparison measure $\rho$ and tolerance level $\varepsilon$.
We may note that the first approximation is exact when $S$ is sufficient.
Allowing the summary statistics to vary in an acceptable range sacrifices
a little accuracy in exchange for a significant improvement in computational
efficiency, which \textcolor{black}{makes } the algorithm more practical
and user-friendly. 

Pursuant to Algorithm \ref{alg:Prichard-ABC}, there are multiple
generalizations in the statistical literature. \cite{marjoram2003markov}
\textcolor{black}{introduced } MCMC-ABC algorithm to concentrate the
samples in high posterior probability region, thereby increasing the
accept rate. Noisy ABC, proposed by \cite{wilkinson2013approximate},
makes use of all the prior samples by assigning kernel weights instead
of hard-threshold accept-reject mechanism and hence reduces the computational
burden. This perspective is corroborated in \cite{fearnhead2012constructing}
by convergence of Bayesian estimators. When the dependence structure
between hierarchies is intractable, ABC filtering technique innovated
by \cite{jasra2012filtering} comes to the rescue. Later in \cite{dean2014parameter},
a consistency argument is established for the specific case of hidden
Markov models. Moreover, many ABC algorithms above can be easily coded
in a parallel way, and hence take advantages of modern computer structures.
This feature makes ABC algorithms extremely time-saving in comparison with long-established,
looping-based MCMC and MLE algorithms. 

Despite the fruitful results on ABC both from applied and theoretical
points of view,  there exist only a handful of papers which
focus on the effect of the choice of summary statistics on the approximation
quality. The quintessential case is when the summary statistics are sufficient,
and the resultant ABC sampler produces exact samples from the true
posterior distribution when $\varepsilon$ goes to zero. Nevertheless,
in a labyrinthine model, it is difficult to extract sufficient statistics,
except for some very special cases, such as exponential random graph
models (e.g. \cite{grelaud2009abc}). \cite{joyce2008approximately}
proposed a concept called $\varepsilon-$sufficient to quantify the
effect of statistics. Nonetheless, this property is also difficult
to verify in complicated models. If we are interested only in model
selection, \cite{prangle2014semi} designed a semi-automatic algorithm
to construct summary statistics via logistic regression. And laterly,
\cite{marin2014relevant} gave sufficient conditions on summary
statistics in order to choose the right model based on the Bayes factors.
They advocated that the ideal summary statistics are ancillary in both
model candidates. One of our contribution comes from the mathematical
analysis of the consequence of conditioning the parameters of interest
on consistent statistics and intrinsically inconsistent statistics,
and appraises the efficiency of the posterior approximation based
on the former. Generally speaking, using consistent statistics results
in right concentration of the approximate posterior, while less efficient
statistics lead to less efficiency of approximation. One byproduct
is our theorem vindicates the usage of the posterior mean as summary
statistics as in \cite{fearnhead2012constructing}. 

In addition to the pure theoretical contribution, we also extend the{{}
}\textcolor{black}{two-step}{{} }algorithm in \cite{fearnhead2012constructing}
in a more flexible and nonparametric way for automatic construction of
summary statistics. We borrow the idea from another thriving topic,
namely sufficient dimension reduction (SDR). The motivation of SDR
which generalizes the concept of sufficient statistics is to estimate{{}
}\textcolor{black}{a} transformation $\varphi$, either{{}
}linear or nonlinear, such that 
\begin{equation}
Y\independent X\mid\varphi\left(X\right).\label{eq:sdr}
\end{equation}
 The first SDR method titled sliced inverse regression dates back
to \cite{li1991sliced}, followed by principle Hessian direction
in \cite{li1992principal} and also by \cite{1991} and \cite{cook1998principal}.
As we step in the era of big data, this idea leads to a sea of papers
on both linear and nonlinear predictors and response. Among the more
recent work, we refer to \cite{cook2002dimension}, \cite{xia2002adaptive},
\cite{li2005contour}, \cite{li2009dimension}, \cite{wu2008kernel},
\cite{yeh2009nonlinear}, \cite{su2011partial} and \cite{su2012inner}.
The association between SDR and ABC relies on the shared mathematical
formulation. If we think $\theta$ as the response and $X$ as the
predictor, then an ideal summary statistics $S\left(X\right)$ will
give 
\[
\theta\independent X\mid S\left(X\right).
\]
This simple observation offers raison d'etre to use existing SDR methods
in constructing summary statistics. The employment of dimension reduction
methods in our algorithm is different from that in \cite{blum2013comparative}.
In \cite{blum2013comparative}, dimension reduction methods, such
as best subset selection, projection techniques and regularization
approaches, are applied to reduce the dimension of existing summary
statistics, but here, we try to reduce the size of the original data.
Particularly in our paper, we incorporate the principal  support vector
machine for nonlinear dimension reduction given in \cite{li2011principal}
into ABC, which uses the{{} }\textcolor{black}{principal
} component of support vectors in reproducing kernel Hilbert space
(RKHS) as a nonparametric estimator of $\varphi$. 

The outline of remaining sections is as follows. Section \ref{sec:asymp-partial-post}
contains asymptotic results on the partial posterior. We gradually
relax the restriction on summary statistics and investigate the relationship
between the partial posterior and the full posterior. As a side result,
we give a lemma building a bridge between the recent prior free inferential
model (\cite{martin2013inferential}, \cite{martin2015conditional})
and traditional Bayesian inference. Section \ref{sec:abc-sdr} elicits
a new ABC algorithm which automatically produces summary statistics
through nonlinear SDR. A simulation result is provided in this section{{}
}\textcolor{black}{as well}. Section \ref{sec:Discussion} briefly
discusses the results and points out some possible future generalizations.

\section{\label{sec:asymp-partial-post}Asymptotic Properties of Partial Posterior}

Suppose $X_{1},\ldots,X_{n}\mid\theta$ are i.i.d. with common PDF
$f\left(x\mid\theta\right)$, and there exists a true but unknown
value $\theta_{0}$. Without loss of generality, we assume $\theta\in\mathbb{R}$,
and all probability density functions are with respect to the Lebesgue
measure. For illustration purpose, we define the following terminology.
\begin{definition}[Partial Posterior]
Let $S=S\left(X_{1},\ldots,X_{n}\right)$ be{{} }\textcolor{black}{statistics
}{{} }of the data. Given a prior $\pi\left(\theta\right)$,
we call the distribution 
\[
\pi\left(\theta\mid S\right)\propto\pi\left(\theta\right)g\left(S\mid\theta\right)
\]
the partial posterior, where $g\left(S\mid\theta\right)$ is the probability
density function of statistic $S\left(X_{1},\ldots,X_{n}\right)$
derived from the data density, and correspondingly, 
\[
\pi\left(\theta\mid X_{1},\ldots,X_{n}\right)\propto\pi\left(\theta\right)f\left(X_{1},\ldots,X_{n}\mid\theta\right)
\]
 is called the full posterior.
\end{definition}
From equation (\ref{eq:two-step-approx-abc}), the partial posterior significantly
reduces the complexity of the full posterior by replacing the dependence
on full data by lower dimensional statistics $S$. If the partial
posterior deviates from \textcolor{black}{the }{{} }full
posterior too much, then no matter how delicately we sample from $\pi\left(\theta\mid S_{\mathrm{sim}}\in O\left(S_{\mathrm{obs}},\varepsilon\right)\right)$
,  and how small $\varepsilon$ we choose, the resultant samples would
not behave like ones drawn from the original full posterior. This{{}
}\textcolor{black}{makes } the subsequent Bayesian analysis fragile
and unreliable. Therefore, theoretical connection between some easily
verifiable properties and asymptotic behaviour of the partial posterior
is of relevance. In particular, we want to study consistency and asymptotic
normality of our Bayesian procedures. The following theorems try to
demonstrate the connection between the asymptotic behaviour of summary
statistics and that of partial posterior. We start from the most popular
statistics, the maximum likelihood estimators (MLE) of $\theta$. 
\begin{theorem}
\label{thm:bernstein-von-mises-mle}Let $\hat{\theta}$, the MLE of
$\theta$, be a strongly consistent estimator, and let $\hat{I}$
be the observed Fisher information evaluated at $\hat{\theta}$, and
the full posterior satisfies the Bernstein--von Mises theorem. Then
for any $\varepsilon>0$, and any $t$, the partial posterior after
conditioned on $\hat{\theta}$ satisfies 
\[
\lim_{n\rightarrow\infty}\prob\left\{ \left(n\hat{I}\right)^{1/2}\left(\theta-\hat{\theta}\right)\le t\mid\hat{\theta}\in O\left(\theta_{0},\varepsilon\right)\right\} =\Phi\left(t\right),\ascv.
\]
\end{theorem}
\begin{proof}
See  Appendix A. \end{proof}
\begin{remark}
There is a slight difference between 
\[
\lim_{n\rightarrow\infty}\prob\left\{ \left(n\hat{I}\right)^{1/2}\left(\theta-\hat{\theta}\right)\le t\mid\hat{\theta}\in O\left(\theta_{0},\varepsilon\right)\right\} =\Phi\left(t\right),\ascv.\left(P_{\theta_{0}}\right)
\]
and 
\[
\lim_{n\rightarrow\infty}\prob\left\{ \left(n\hat{I}\right)^{1/2}\left(\theta-\hat{\theta}\right)\le t\mid\hat{\theta}\right\} =\Phi\left(t\right),\ascv.
\]
By definition, 
\begin{equation}
\prob\left\{ \left(n\hat{I}\right)^{1/2}\left(\theta-\hat{\theta}\right)\le t\mid\hat{\theta}\right\} =\lim_{\varepsilon\rightarrow0}\frac{\prob\left\{ \left(n\hat{I}\right)^{1/2}\left(\theta-\hat{\theta}\right)\le t,\hat{\theta}\in O\left(s,\varepsilon\right)\right\} }{\prob\left\{ \hat{\theta}\in O\left(s,\varepsilon\right)\right\} }.\label{eq:def-loose-partial-post}
\end{equation}
The result of Theorem \ref{thm:bernstein-von-mises-mle} can only
be used to prove 
\[
\lim_{\varepsilon\rightarrow0}\lim_{n\rightarrow\infty}\prob\left\{ \left(n\hat{I}\right)^{1/2}\left(\theta-\hat{\theta}\right)\le t\mid\hat{\theta}\in O\left(\theta_{0},\varepsilon\right)\right\} =\Phi\left(t\right),\ascv,
\]
switching order of limits in equation (\ref{eq:def-loose-partial-post}). 
\begin{remark}
The definition of $\prob\left\{ \theta\mid\hat{\theta}\in O\left(\theta_{0},\varepsilon\right)\right\} $
is different from the approximation $\prob\left\{ \theta\mid\hat{\theta}\in O\left(\hat{\theta}_{\mathrm{obs}},\varepsilon\right)\right\} $.
In first case, $\hat{\theta}$ is evaluated at $X_{1},\ldots,X_{n}\sim f\left(x\mid\theta_{0}\right)$,
the observed data, while the latter evaluates $\hat{\theta}$ at $Z_{1},\ldots,Z_{m}\sim f\left(z\mid\theta\right)$,
the simulated data. 
\end{remark}
\end{remark}
By{{} }\textcolor{black}{assumptions }, the asymptotic
distribution of the full posterior is still normal, and we have 
\begin{eqnarray*}
 &  & \sup_{t\in\mathbb{R}}\Bigg|\prob\left\{ \left(n\hat{I}\right)^{1/2}\left(\theta-\hat{\theta}\right)\le t\mid\hat{\theta}\in O\left(\theta_{0},\varepsilon\right)\right\} \\
 &  & -\prob\left\{ \left(n\hat{I}\right)^{1/2}\left(\theta-\hat{\theta}\right)\le t\mid X_{1},\ldots,X_{n}\right\} \Bigg|\\
 & \le & \sup_{t\in\mathbb{R}}\left|\prob\left\{ \left(n\hat{I}\right)^{1/2}\left(\theta-\hat{\theta}\right)\le t\mid\hat{\theta}\in O\left(\theta_{0},\varepsilon\right)\right\} -\Phi\left(t\right)\right|\\
 &  & +\sup_{s\in\mathbb{R}}\left|\prob\left\{ \left(n\hat{I}\right)^{1/2}\left(\theta-\hat{\theta}\right)\le s\mid X_{1},\ldots,X_{n}\right\} -\Phi\left(s\right)\right|\rightarrow0,\:\left(\mathrm{as}\: n\rightarrow\infty\right).
\end{eqnarray*}
Hence, we can informally say that two random variables $\left(n\hat{I}\right)^{1/2}\left(\theta-\hat{\theta}\right)\mid\hat{\theta}$
and $\left(n\hat{I}\right)^{1/2}\left(\theta-\hat{\theta}\right)\mid X_{1},\ldots,X_{n}$
are close in distribution. Note that both random variables asymptotically
center at consistent MLE, and hence will eventually concentrate at
$\theta_{0}$. Meanwhile, the scale factors in both random variables
are $\left(n\hat{I}\right)^{1/2}$, which{{} }\textcolor{black}{ensures
}{{} }the same square root credible \textcolor{black}{intervals}.
In this sense, we feel that the partial posterior conditioned on the MLE
has the same efficiency as the full posterior. Later theorems will
tell us that if the summary statistics are not efficient, the corresponding
partial likelihood will have a different scale factor, and thus will lose
efficiency and result in a larger credible interval. 

A slightly modified proof of Theorem \ref{thm:bernstein-von-mises-mle}
can be used to support the posterior mean as a summary statistic in
\cite{fearnhead2012constructing} and we still have a similar result,
namely 
\begin{eqnarray*}
 &  & \lim_{n\rightarrow\infty}\prob\Bigg[\left(n\hat{I}\right)^{1/2}\left\{ \theta-E\left(\theta\mid X_{1},\ldots,X_{n}\right)\right\} \le t\\
 &  & \mid E\left(\theta\mid X_{1},\ldots,X_{n}\right)\in O\left(\theta_{0},\varepsilon\right)\Bigg]=\Phi\left(t\right),\ascv.
\end{eqnarray*}
The key fact to support the assertion above comes from \cite{ghosh2011moment},
that is, the higher order closeness of the posterior mean and the
MLE, namely 
\begin{equation}
\lim_{n\rightarrow\infty}n^{1/2}\left\{ E\left(\theta\mid X_{1},\ldots,X_{n}\right)-\hat{\theta}\right\} =0,\ascv.\label{eq:high-order-close-post-mean-mle}
\end{equation}
Indeed, any estimator who has the same or higher order of closeness
to MLE will work as an efficient summary statistic.

Theorem \ref{thm:bernstein-von-mises-mle} can be generalized to more
intricate models. The following example shows the same phenomenon
in data generated from a Markov process.
\begin{example}
\label{exa:Immigrate-emigrate-process}Immigration-emigration process
is{{} }\textcolor{black}{a }{{} }crucial
model in survival analysis and can be viewed as a special case of
mass-action stochastic kinetic{{} }\textcolor{black}{network
} (\cite{wilkinson2011stochastic}). The model is defined by a birth
procedure and{{} }\textcolor{black}{a } {{} }death
procedure during an infinitesimal time interval, namely, 
\[
\prob\left\{ X\left(t+\diff t\right)=x_{1}\mid X\left(t\right)=x_{0}\right\} =\begin{cases}
\lambda\diff t+o\left(\diff t\right), & x_{1}=x_{0}+1,\\
\mu x_{0}\diff t+o\left(\diff t\right), & x_{1}=x_{0}-1,\\
1-\lambda\diff t-\mu x_{0}\diff t+o\left(\diff t\right), & x_{1}=x_{0}.
\end{cases}
\]
Assume that we observe full data in the time interval $\left[0,T\right]$.
Let $T_{i},i=1,\ldots,n$ be the event times and let $X_{i}=X\left(T_{i}\right),i=1,\ldots,n$.
Let $X_{0}$ be initial population, $T_{0}=0$, $T_{n+1}=T$. Then
by Gillespie's algorithm, the likelihood is proportional to 
\[
\lambda^{r_{1}}\exp\left(-\lambda T\right)\mu^{r_{2}}\exp\left(-\mu A_{T}\right),
\]
where $r_{1}$ and $r_{2}$ are number of events corresponding to
immigration and emigration, and 
\[
A_{T}=\int_{0}^{T}X\left(t\right)\diff t.
\]
The MLEs are 
\[
\hat{\lambda}=\frac{r_{1}}{T},\hat{\mu}=\frac{r_{2}}{A_{T}},
\]
and they are strongly consistent estimators of $\lambda$ {{}
}\textcolor{black}{and $\mu$} {{} }when $T$ goes to
infinity. By the computation in  Appendix B.1,
we have the partial posterior density function of $T^{1/2}\left(\mu-\hat{\mu}\right)$
conditioned on $\hat{\mu}$, $r_{1}$ and $T$ given by 
\[
\lim_{T\rightarrow\infty}\pi\left\{ T^{1/2}\left(\mu-\hat{\mu}\right)=t\mid\hat{\mu},r_{1},T\right\} =\frac{\hat{\mu}}{\left(2\pi\hat{\lambda}\right)^{1/2}}\exp\left(-\frac{\hat{\lambda}}{\hat{\mu}^{2}}t^{2}\right),\ascv.
\]

\end{example}
The MLE seems to be a perfect surrogate for \textcolor{black}{the
} {{} }full data. However, in many cases, use of MLE is
prohibitive due to heavy computational burden, particularly when the
likelihood function is intractable. This is when the ABC comes on
stage. $M$-estimator is a generalization of the MLE, which is also
consistent and asymptotically normal under mild conditions. Many $M$-estimators
can be easily calculated, especially some moment estimators. To give
an idea of the nature of approximation, we consider the following
examples.
\begin{example}
\label{exa:Gamma-distribution}Gamma distribution can be used to model
hazard functions in survival analysis. The shape parameter of gamma
distribution determines the trend of hazard and hence is a vital parameter
to estimate. Assume $X_{1},\ldots,X_{n}\sim\mathrm{Gamma}\left(\alpha,\beta\right)$,
where we know the scale parameter $\beta$, but not the shape parameter
$\alpha$. The MLE of $\alpha$ is the solution of 
\[
-\log\Gamma\left(\alpha\right)-\alpha\log\beta+\left(\alpha-1\right)\sum_{i=1}^{n}\log X_{i}-\frac{\sum_{i=1}^{n}X_{i}}{\beta}=0,
\]
which involves repeated evaluation of the gamma function in search
of the root. A simple $M$-estimator $\tilde{\alpha}=\overline{X}/\beta$
is derived from its mean equation, 
\[
\sum_{i=1}^{n}\left(X_{i}-\alpha\beta\right)=0.
\]
 Now we consider the partial posterior $\pi\left(\alpha\mid\tilde{\alpha}\right)$,
when the prior is $\pi\left(\alpha\right)\propto\exp\left(-\lambda\alpha\right)$.
By the calculation in  Appendix B.2, we show
that the limit of cumulative probability function of $n^{1/2}\tilde{\alpha}^{-1}\left(\alpha-\tilde{\alpha}\right)$
given $\tilde{\alpha}$ is 
\[
\lim_{n\rightarrow\infty}\prob\left\{ n^{1/2}\tilde{\alpha}^{-1}\left(\alpha-\tilde{\alpha}\right)\le t\mid\tilde{\alpha}\right\} =\Phi\left(t\right),\ascv,
\]
which means that the Bernstein-von Mises theorem holds for the partial
posterior conditioned on the $M$-estimator $\tilde{\alpha}$. The
scale factor of \textcolor{black}{{} the} partial posterior is $n^{1/2}\tilde{\alpha}^{-1}$,
which is smaller than that of {{} }\textcolor{black}{the
} {{} }full posterior, $\left\{ n\psi'\left(\alpha\right)\right\} ^{1/2}$,
where $\psi\left(\alpha\right)$ is digamma function. That results
in a larger credible interval based on the partial posterior. 
\begin{example}
\label{exa:laplace-example}Another example is the Laplace distribution
with pdf 
\[
f_{\mu,\lambda}\left(t\right)=\frac{1}{2\lambda}\exp\left(-\frac{\left|t-\mu\right|}{\lambda}\right).
\]
 Here we want inference for the location parameter $\mu$ holding
$\lambda$ fixed. The MLE is the {{} }\textcolor{black}{the
} sample median and the moment estimator is {{} }\textcolor{black}{the
}sample mean. Here we calculate the partial posterior based on the sample
mean. By the calculation in  Appendix B.3, we
find that the characteristic function of $n^{1/2}\left(\mu-\overline{X}\right)$
converges to $\exp\left(-\lambda^{2}t^{2}\right)$, which is the characteristic
function of normal distribution. 
\end{example}
\end{example}
Example \ref{exa:laplace-example} uses the following lemma which
is of independent interest. 
\begin{lemma}
\label{lem:bayes-inferential-model}Assume $X$ has the same distribution
as $h\left(Y,\theta\right)$, where $h\left(y,\theta\right)$ for
a fixed $y$ is a one-to-one function of $\theta$ and $Y$ is a random
variable independent of $\theta$. Let $\theta=g\left(y,x\right)$
and $y=u\left(x,\theta\right)$ be the solutions of the equation $x=h\left(y,\theta\right)$.
Further assume $\partial u\left(x,\theta\right)/\partial x$ exists
and is not equal to zero. Then the posterior distribution of $\theta$
conditioned on $X$ under the uniform prior has the same distribution
as $g\left(Y,x\right)$, where $x$ is fixed. \end{lemma}
\begin{remark}
Although not quite related to ABC, this lemma gives another interpretation
of inferential model of \cite{martin2013inferential} and \cite{martin2015conditional}.
In their settings, $Y$ is called unobserved ancillary variable, and
$g\left(y,x\right)$ is $\Theta_{x}\left(u\right)$ in their notation.
They claim that their procedure results in a distribution of $\theta$
without referring to a prior. However, by our lemma, this model is
mathematically the same as a posterior given a uniform prior.
\end{remark}
The following theorems are built upon the Theorem 2.1 in \cite{rivoirard2012bernstein},
which guarantees the asymptotic normality of linear functionals of
nonparametric posterior. So we need all the assumptions in that theorem.
Additionally, we need the following assumptions. 
\begin{assumption}
\label{assu:second-order-bounded-differential}There is a neighbourhood
$\theta\in O\left(\theta_{0},\varepsilon\right)$ such that $\int_{\mathbb{R}}g\left(x,\theta_{0}\right)\pi\left(x\mid\theta\right)\diff x$
is a continuous twice differentiable in $\theta$ and the second order
derivative is bounded by some constant $L$. 
\begin{assumption}
\label{assu:m-est-consistent-asymp-norml}$M$-estimator $\tilde{\theta}$
and MLE $\hat{\theta}$ are both strongly consistent and asymptotically
normal.
\begin{assumption}
\label{assu:bernstein-von-mises-full-posterior} Bernstein--von Mises
theorem and posterior consistency hold for the full posterior of $\theta$.
\begin{assumption}
\label{assu:theo-mle}For any $\theta\in\Theta$, $E_{\theta_{0}}\log f\left(X\mid\theta\right)\le E_{\theta_{0}}\log f\left(X\mid\theta_{0}\right)$. 
\end{assumption}
\end{assumption}
\end{assumption}
\end{assumption}
Now we can articulate the theorem. 
\begin{theorem}
\label{thm:partial-post-m-est}Under the Assumptions \ref{assu:second-order-bounded-differential},
\ref{assu:m-est-consistent-asymp-norml}, \ref{assu:bernstein-von-mises-full-posterior},
\ref{assu:theo-mle}, and conditions of Theorem 2.1 in \cite{rivoirard2012bernstein},
for any $\varepsilon$ and $t$, 
\[
\lim_{n\rightarrow\infty}\prob\left\{ \left.\left(n/\tilde{V}\right)^{1/2}\left(\theta-\tilde{\theta}\right)\le t\right|\tilde{\theta}\in O\left(\theta_{0},\varepsilon\right)\right\} =\Phi\left(t\right),\ascv,
\]
where $\tilde{V}=V_{0}/G_{1}\left(\tilde{\theta},\tilde{\theta}\right)^{2}$
is the Godambe information. \end{theorem}
\begin{proof}
See  Appendix C.
\end{proof}
Using similar arguments as Theorem \ref{thm:bernstein-von-mises-mle},
the partial posterior $\left(n/\tilde{V}\right)^{1/2}\left(\theta-\tilde{\theta}\right)\mid\tilde{\theta}$ {{}
}\textcolor{black}{is asymptotically } close in distribution to \textcolor{black}{the
}full posterior $\left(n\hat{I}\right)^{1/2}\left(\theta-\hat{\theta}\right)\mid X_{1},\ldots,X_{n}$.
Since both the $M$-estimator and the MLE are strongly consistent,
the partial posterior still concentrates around the right $\theta_{0}$,
but now the asymptotic $\alpha-$level credible interval based on
the partial posterior, namely 
\[
\left(\tilde{\theta}-Z_{\alpha/2}\left(\frac{\tilde{V}}{n}\right)^{1/2},\tilde{\theta}+Z_{\alpha/2}\left(\frac{\tilde{V}}{n}\right)^{1/2}\right),
\]
will be larger than that based on the full posterior, 
\[
\left(\hat{\theta}-Z_{\alpha/2}\left(\frac{\hat{I}^{-1}}{n}\right)^{1/2},\hat{\theta}+Z_{\alpha/2}\left(\frac{\hat{I}^{-1}}{n}\right)^{1/2}\right),
\]
where $Z_{\alpha/2}$ is the $\left(1-\alpha/2\right)$ quantile of standard
normal distribution. This is because the Godambe information $\tilde{V}^{-1}$
is typically no larger than Fisher information $\hat{I}$. Hence,
we lose efficiency if we condition the posterior on an inefficient
estimator, which coincides with our intuition.

For extreme tortuous models, even finding a consistent estimator can
be quite hard. There are still some simple statistics which may be
consistent to some {{} }\textcolor{black}{functions } of
$\theta$. Unless they are ancillary statistics, they always contain
some information about the {{} }\textcolor{black}{parameters
} of interest. Moreover, in the real case, we use several statistics,
each of which gives independent information of the full posterior.
In the remainder of this section, we will mathematically quantify what
independent information means and show that using more than one statistic
will improve the efficiency. 

Let $S_{i}$, $i=1,\ldots,q$ be statistics of the sample. We make
the following trivial assumptions. 
\begin{assumption}
\label{assu:joint-normal-inconsist-stat}The joint distribution of
$S_{1},\ldots,S_{q}$ converges in distribution to a multivariate
normal distribution $N\left(h\left(\theta_{0}\right),n^{-1/2}\Sigma\left(\theta_{0}\right)\right)$,
and each $S_{i}$ converges to $h_{i}\left(\theta_{0}\right)$ almost
surely. Further, assume $\Sigma\left(\theta_{0}\right)$ {{}
}\textcolor{black}{is } positive definite, and $h\left(\theta_{0}\right)$
is a linear functional of the distribution function, that is 
\[
h\left(\theta_{0}\right)=\int_{\mathbb{R}}g\left(x\right)f\left(x\mid\theta_{0}\right)\diff x,
\]
where $g\left(x\right)\in\mathbb{R}^{q}$. 
\end{assumption}
Assumption \ref{assu:joint-normal-inconsist-stat} characterizes the
 independent information statement. Because if $\Sigma\left(\theta_{0}\right)$
has a lower rank, then some of $S_{i}$ can be expressed as \textcolor{black}{
}linear combinations of {{} }\textcolor{black}{other } $S_{j}$
asymptotically. Then the partial posterior can be reduced to a partial
posterior based solely on the $S_{j}$. The functional form of $h$ is
a natural consequence when we apply some version of strong law of
large numbers to prove convergence of statistics. 

In order to prove the theorem, we need some more technical assumptions.
\begin{assumption}
\label{assu:super-strong-consistent}Let $S=\left(S_{1},\ldots,S_{q}\right)$,
assume 
\[
\lim_{n\rightarrow\infty}n^{1/2}\left\{ \frac{1}{n}\sum_{i=1}^{n}g\left(X_{i}\right)-S\right\} =0,\ascv.
\]
and there exists a strongly consistent estimator $\tilde{\Sigma}$
of $\Sigma\left(\theta_{0}\right)$
\end{assumption}
Only Assumption \ref{assu:super-strong-consistent} seems quite restrictive.
Based on all these assumptions, the theorem describing the partial
posterior conditioned on less informative statistics can be found
as follows:
\begin{theorem}
\label{thm:bernsten-von-mise-inconsist-multv}Under Assumptions \ref{assu:joint-normal-inconsist-stat},
 \ref{assu:super-strong-consistent} and conditions of Theorem 2.1
in \cite{rivoirard2012bernstein}, for any vector $a\in\mathbb{R}^{q}$,
\[
\lim_{n\rightarrow\infty}\sup_{t\in\mathbb{R}}\left|\prob\left[\left.\frac{n^{1/2}a^{T}\left\{ h\left(\theta\right)-S\right\} }{\left(a^{T}\tilde{\Sigma}a\right)^{1/2}}\le t\right|S\right]-\Phi\left(t\right)\right|=0,\ascv.
\]
\end{theorem}
\begin{proof}
See Appendix D.
\end{proof}
Theorem \ref{thm:bernsten-von-mise-inconsist-multv} extends the asymptotic
results about $M$-estimators to more general statistics, particularly
the intrinsically inconsistent statistics defined as follows.
\begin{definition}[Intrinsic Consistency]
Let $S$ be an non-ancillary statistic and converges to $h\left(\theta_{0}\right)$
almost surely. If $h\left(\cdot\right)$ is an one-to-one function
and has an inverse function, then we say $S$ is intrinsically consistent.
Otherwise, we say $S$ is intrinsically inconsistent. 
\end{definition}
If $S$ is a one dimensional intrinsically inconsistent statistic, then
Theorem \ref{thm:bernsten-von-mise-inconsist-multv} asserts the $\left(1-\alpha\right)$
asymptotic credible set based on {{} }\textcolor{black}{the
} partial posterior is 
\[
\left\{ \theta:S-Z_{\alpha/2}\left(\frac{\tilde{\Sigma}}{n}\right)^{1/2}\le h\left(\theta\right)\le S+Z_{\alpha/2}\left(\frac{\tilde{\Sigma}}{n}\right)^{1/2}\right\} .
\]
In an extreme case, when sample size $n$ is large enough, such that
$Z_{\alpha/2}/\sqrt{n}\approx0$, the asymptotic credible interval
based on the {{} }\textcolor{black}{the } full posterior would be
close to the singleton $\left\{ \hat{\theta}\right\} $. However,
the credible set based on {{} }\textcolor{black}{the } partial
posterior would be $\left\{ \theta:h\left(\theta\right)=S\right\} $.
By the definition of intrinsic inconsistency, $h$ is not a one-to-one
function. Then the set $\left\{ \theta:h\left(\theta\right)=S\right\} $
would possibly hold multiple elements, hence larger than that from {{}
}\textcolor{black}{the } full posterior. Again, in this sense, we
perceive loss of efficiency due to conditioning the posterior on arbitrary
statistics. 

Another interesting use of Theorem \ref{thm:bernsten-von-mise-inconsist-multv}
is a more pragmatic asymptotic assessment of effectiveness of including {{}
}\textcolor{black}{many} statistics than that in \cite{joyce2008approximately}.
In their settings, the effectiveness of summary statistics is measured
by the difference between log-likelihoods, thus not operable when
likelihood \textcolor{black}{{} functions are} intractable. On the other
hand, our approach only {{} }\textcolor{black}{requires
}the asymptotic behaviour of statistics, and the corresponding credible
set with $q$ statistics can be developed by the Cramer device as 
\[
\left\{ \theta:n\left(h\left(\theta\right)-S\right)^{T}\tilde{\Sigma}\left(h\left(\theta\right)-S\right)\le\chi_{1-\alpha,q}^{2}\right\} ,
\]
where $\chi_{1-\alpha,q}^{2}$ is $\left(1-\alpha\right)$ quantile
of chi-square distribution with degree of freedom $q$. To select
summary statistics, we can compare the asymptotic credible sets with
and without the current statistic. If the difference is small, then
we can safely throw the current statistic away. 

The following is a simple example to illustrate this phenomenon.
\begin{example}
\textcolor{black}{Let $X_{i}$ are i.i.d. sample from $N\left(\mu,\mu^{2}\right)$
and we calculate the partial posterior $\pi\left(\mu\mid s^{2}\right)$,
where $s^{2}$ is the sample standard deviation.  The prior of $\mu^{2}$
is assumed to be inverse gamma distribution with shape parameter $\alpha$
and scale parameter $\beta$. We know the $s^{2}\sim\mu^{2}\chi_{n-1}^{2}/\left(n-1\right)$,
hence the partial posterior of $\mu^{2}$ conditioning on $s^{2}$
is an inverse gamma distribution with shape parameter $\left\{ \alpha-1+\left(n-1\right)/2\right\} $
and scale parameter $\left\{ \beta+\left(n-1\right)s^{2}/2\right\} $.
This leads to a bimodal partial posterior for $\mu$. Hence, we can
only get the absolute value of $\mu$ without sign information from
this partial posterior.}
\end{example}

\section{\label{sec:abc-sdr}Approximate Bayesian Computation {{}
}\textcolor{black}{via } Nonlinear Sufficient Dimension Reduction}

In principle, almost all the existing dimension reduction methods
are valid in estimating the summary statistics. However, there is
a slight difference between the setting of SDR and ABC. In the theory
of SDR, the independent assumption \ref{eq:sdr} must hold rigorously,
which implies $Y\mid X$ has exactly the same distribution as $Y\mid S\left(X\right)$.
However, by our Theorem \ref{thm:bernstein-von-mises-mle}, \ref{thm:partial-post-m-est}
and \ref{thm:bernsten-von-mise-inconsist-multv}, the two distributions
are only close in large but finite samples.

\subsection{Algorithm: ABC via PSVM}

In our paper, we choose principal support vector machine in \cite{li2011principal}.
Suppose we have a regression problem $\left(Y_{i},X_{i}\right)$,
and search a nonlinear transformation $\varphi:\mathbb{R}^{p}\rightarrow\mathbb{R}^{d}$,
such that $Y\independent X\mid\varphi\left(X\right)$. Then the main
steps in principal support vector machine are given in Algorithm \ref{alg:Principal-Support-Vector}
. 

\begin{algorithm}
\begin{enumerate}
\item (Optional) Marginally standardize data $X_{1},\ldots,X_{n}$. The
purpose of this step is so that the kernel $\kappa$ treats different
components of $X_{i}$ more or less equally. 
\item Choose kernel $\kappa$ and the number of basis {{} }\textcolor{black}{functions
}$k$ (usually around $n/3\sim2n/3$). Compute $K=\left\{ \kappa\left(X_{i},X_{j}\right)\right\} _{n\times n}$.
Let $Q=I_{n}-J_{n}/n$, where $J_{n}$ is the $n\times n$ matrix
whose entries are 1. Compute largest $k$ eigenvalues $\lambda_{1},\ldots,\lambda_{k}$
and corresponding eigenvectors $w_{1},\ldots,w_{k}$ of matrix $QKQ$.
Let $\Psi=\left(w_{1},\ldots,w_{k}\right)$ and $P_{\Psi}=\Psi\left(\Psi^{T}\Psi\right)^{-1}\Psi^{T}$
be the corresponding projection matrix.
\item \label{enu:svm}Partition the response variable space $Y$ into $h$
slices defined by $y_{1},\ldots,y_{h-1}$. For each $y_{s},\: s=1,\ldots,h-1$,
define a new response variable $\tilde{Y}_{si}=I_{\left[Y_{i}\le y_{s}\right]}-I_{\left[Y_{i}>y_{s}\right]}$.
Then solve the modified support vector machine problem as a standard
quadratic programming 
\[
\min_{\alpha}-1^{T}\alpha+\frac{1}{4}\alpha^{T}\mathrm{diag}\left(\tilde{Y}_{s}\right)P_{\Psi}\mathrm{diag}\left(\tilde{Y}_{s}\right)\alpha,
\]
subject to constraints
\[
\begin{cases}
0\le\alpha\le\lambda,\\
\tilde{Y}_{s}^{T}\alpha=0,
\end{cases}
\]
where $\mathrm{diag}\left(\tilde{Y}_{s}\right)$ is a diagonal matrix
using $\tilde{Y}_{s}$ as diagonal, $\lambda$ is a hyper-parameter
in ordinary support vector machine. The coefficients of support vectors
in RKHS {{} }\textcolor{black}{are} 
\[
c_{s}^{*}=\frac{1}{2}\left(\Psi^{T}\Psi\right)^{-1}\Psi^{T}\mathrm{diag}\left(\tilde{Y}_{s}\right)\alpha_{s}.
\]

\item \label{enu:pca}Let $d$ be the target dimension. Compute the eigenvectors
$v_{1},\ldots,v_{d}$ of first largest $d$ eigenvalues of the matrix
$\sum_{s=1}^{h-1}c_{s}^{*}c_{s}^{*T}$. Let $V=\left(v_{1},\ldots,v_{d}\right).$
\item Let $K\left(x,X\right)=\left\{ \kappa\left(x,X_{i}\right)-n^{-1}\sum_{j=1}^{n}\kappa\left(x,X_{j}\right)\right\} $
be a $n$ dimensional vector. Then the estimated transformation $\hat{\varphi}\left(x\right)=V^{T}\left\{ \mathrm{diag}\left(\lambda_{1},\ldots\lambda_{k}\right)\right\} ^{-1}\Psi^{T}K\left(x,X\right).$
\end{enumerate}
\protect\caption{Principal Support Vector Machine\label{alg:Principal-Support-Vector}}

\end{algorithm}
By slicing the response variable space, we discretize variation of
$Y$. The support vector machine in Step \ref{enu:svm} recognizes
the robust separate hyperplanes. We will expect the variation of $Y$
along the directions within hyperplanes to be negligible and that
along the directions perpendicular to the hyperplanes to explain the
most part of covariation between $Y$ and $X$. The principal component
analysis on the support vectors in Step \ref{enu:pca} estimates the {{}
}\textcolor{black}{principal } {{} }perpendicular directions
and hence creates the sufficient directions in RKHS. 

Based on Algorithm \ref{alg:Principal-Support-Vector}, we formulate
our \textcolor{black}{two-step }approximate Bayesian computation algorithm
in Algorithm \ref{alg:ABC-via-PSVM}.

\begin{algorithm}

\begin{enumerate}
\item \label{enu:direct-sampling}Sample $\theta_{i}$ from the prior $\pi\left(\theta\right)$
and sample $X_{i1},\ldots,X_{in}$ from the model $f\left(x\mid\theta_{i}\right)$. 
\item View $\left(\theta_{i},X_{i1},\ldots X_{in}\right)$ as a multivariate
regression problem and reduce the dimension from $n$ to $d$ via
principal support vector machine. Denote the estimated transformation
as $\hat{S}\left(X_{1},\ldots,X_{n}\right)$. 
\item Either use existent samples in Step \ref{enu:direct-sampling} or
repeat it and get new sample. Calculate the estimated summary statistics
$\hat{S}_{i}=\hat{S}\left(X_{i1},\ldots,X_{in}\right)$ on each set
$X_{i1},\ldots,X_{in}$ corresponding to prior samples $\theta_{i}$.
Also calculate $\hat{S}_{\mathrm{obs}}=\hat{S}\left(X_{1},\ldots,X_{n}\right)$
on the observed data set.
\item Based on the metric $\rho\left(\hat{S}_{i},\hat{S}_{\mathrm{obs}}\right)$,
make the decision of accept or reject of $\theta_{i}$. 
\end{enumerate}
\protect\caption{\label{alg:ABC-via-PSVM}ABC via PSVM}

\end{algorithm}
Algorithm \ref{alg:ABC-via-PSVM} directly generalizes the semi-automatic
ABC in \cite{fearnhead2012constructing}. In their algorithm, the
summary statistics are fixed as posterior means and the recommended
estimation method is polynomial regression. Our algorithm relaxes
the restriction on summary statistics and {{} }\textcolor{black}{lets
}the data and nonparametric algorithm together find them adaptively.
One significant difference between our algorithm and the conventional
ABC is in Step \ref{enu:direct-sampling}, where each prior sample
$\theta_{i}$ produces exact $n$ simulated data, because the nonparametric
estimator of statistics should take $n$ arguments so that it can
be evaluated at both observed data and simulated data.

\subsection{Robustness of PSVM}

\textcolor{black}{The discrepancy between the asymptotic behavior
of ABC and settings in SDR requires new properties of PSVM, namely
robustness, which suggests that $\hat{\Gamma}$ from PSVM would be
in the vicinity of the true summary statistics function $\Gamma$
in some sense even if the partial posterior is only close to the full
posterior. Here is the theorem which validates this property.}
\begin{theorem}[Robustness of PSVM]
\textcolor{black}{\label{thm:robust-psvm} Assume 
\[
\lim_{n\rightarrow\infty}\sup_{s}\left|\prob\left\{ \theta\le s\mid\Gamma_{n}^{T}\left(X_{1},\ldots,X_{n}\right)\right\} -\prob\left(\theta\le s\mid X_{1},\ldots,X_{n}\right)\right|=0,
\]
where $\Gamma\in\mathbb{R}^{n\times d}$, with $d$ fixed. Further
assume all the conditions in Theorem 4 and 5 in \cite{li2011principal}. Let $\hat{\Gamma}_{n}$
be the result from PSVM. Then 
\[
\hat{\Gamma}_{n}-\Gamma_{n}\overset{p}{\rightarrow}0.
\]
}\end{theorem}
\begin{proof}
\textcolor{black}{See Appendix E. }
\end{proof}

\subsection{\textcolor{black}{Simulation Example }}

\textcolor{black}{First, we will show a simple simulation example
to illustrate the robustness of SVM, which serves as the foundation
of the robustness of PSVM.}
\begin{example}
\textcolor{black}{Let $\theta\mid X_{1},X_{2}\sim N\left(2X_{1}+X_{2}+0.001\left(X_{1}^{2}+X_{2}^{2}\right),1\right)$.
Hence the conditional distribution of $\theta$ on $X_{1},X_{2}$
is close to $N\left(2X_{1}+X_{2},1\right)$. Then the normal vectors
of sliced SVM should be in vicinity of $\left(2,1\right)$. We choose
the slicing point $s=1.5$. The results are summarized in Fig \ref{fig:Normal-Vectors-sliced-svm}.
The dots are normal vectors, the solid line is the principal component
direction of normal vectors, and the reference dashed line is $\psi_{1}=2\psi_{2}$.
We can see the principal component direction is quite close to the
reference line.}

\begin{figure}
{\includegraphics[scale=0.5]{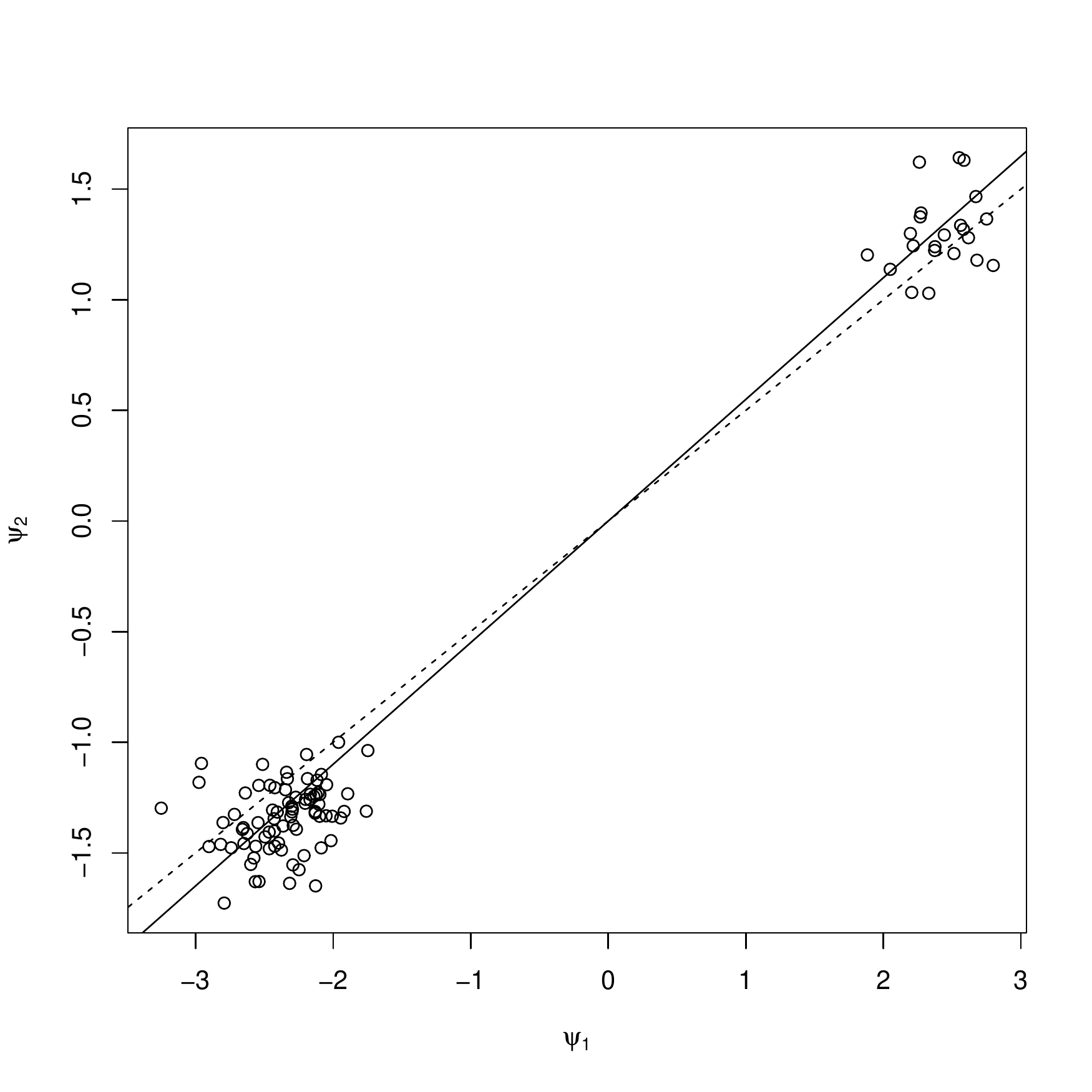}\protect\caption{\textcolor{black}{Normal Vectors of Separate Hyperplanes in Sliced
SVM}\label{fig:Normal-Vectors-sliced-svm}}
}

\end{figure}

\end{example}
Next, we will show a simple simulation example to illustrate our algorithm.
\begin{example}
Autoregressive model with lag one, AR(1). 
\[
Y_{i}=\beta Y_{i-1}+\varepsilon.
\]
Set $Y_{1}=1$ and number of observation is 100. Assume $\varepsilon\sim N\left(0,0.5^{2}\right)$,
true regression coefficient 0.6. We put uniform prior in $\left(-1,1\right)$
on $\beta$. Then the true posterior distribution is $N\left(\sum_{i=1}^{99}Y_{i}Y_{i+1}/\left(1+\sum_{i=1}^{99}Y_{i}^{2}\right),\left(1+\sum_{i=1}^{99}Y_{i}^{2}\right)^{-1}\right)$.
Now we apply our algorithm with the target dimension $d=1$ and slicing
pieces $h=4$ with the slicing parameters $y_{k}$ as quartiles. The
sample size from the prior is 1000, with $k=100/2=500$. Kernel $\kappa$
is chosen as Gaussian kernel $\kappa\left(x_{i},x_{j}\right)=\exp\left(-10^{-5}\times\|x_{i}-x_{j}\|^{2}\right)$.
Then the posterior density estimated from ABC samples are plotted
in Fig. \ref{fig:ABC-vs-True}. %

\begin{figure}

\includegraphics[scale=0.5]{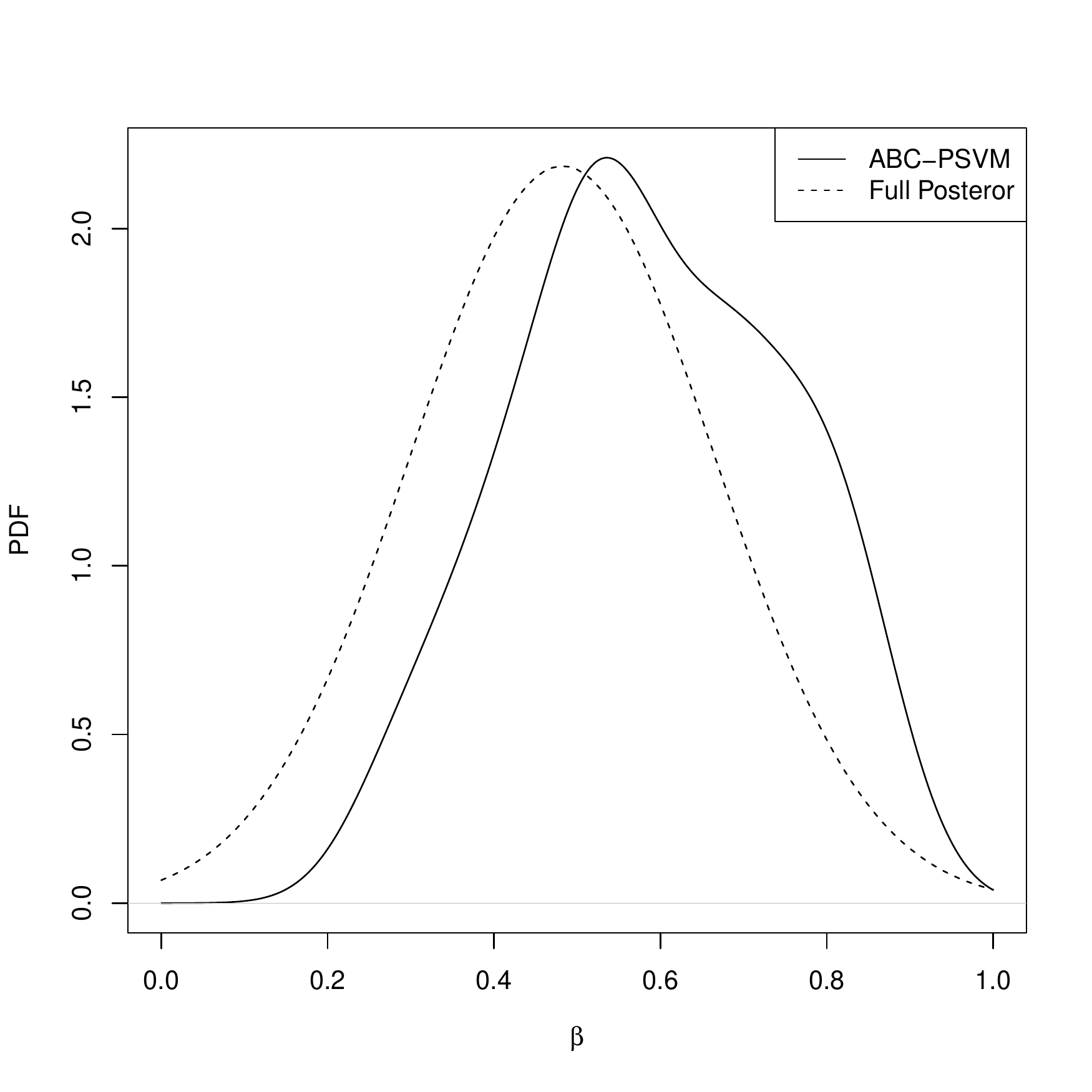}\protect\caption{ABC Density vs True Posterior Density\label{fig:ABC-vs-True}}

\end{figure}
The slight skewness in Fig. \ref{fig:ABC-vs-True} possibly due the
the small sample size of the observed data. Another interesting result
of this simulation is shown in Fig. \ref{fig:ss-vs-mle}. 
\begin{figure}

\includegraphics[scale=0.5]{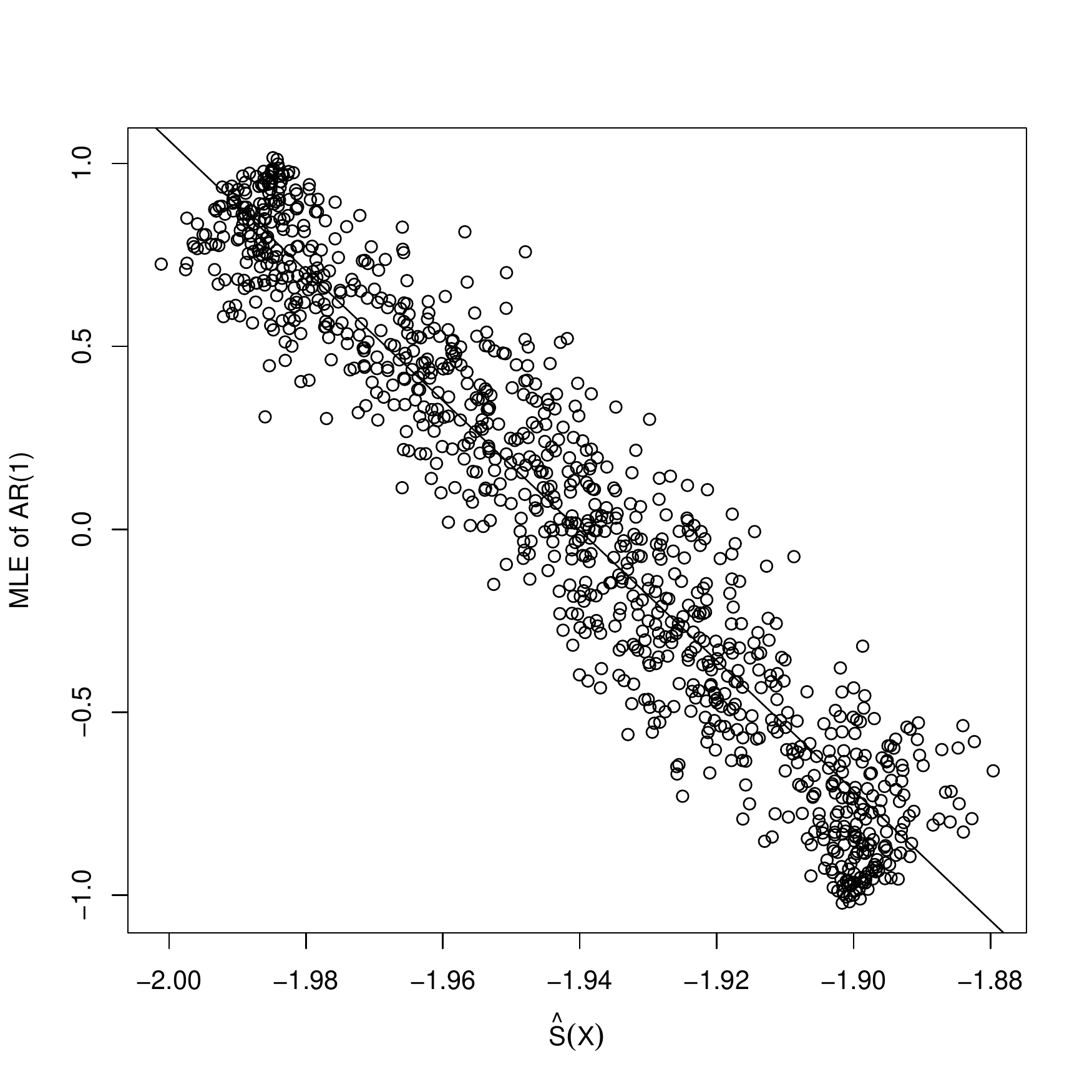}\protect\caption{Estimated Summary Statistic vs MLE\label{fig:ss-vs-mle}}

\end{figure}
There is a strong linear relationship between the estimated summary
statistic and MLE 
\[
\hat{\beta}=\frac{\sum_{i=1}^{99}Y_{i}Y_{i+1}}{\sum_{i=1}^{99}Y_{i}^{2}},
\]
which is one of the most efficient summary{{} }\textcolor{black}{statistics
} based on Theorems \ref{thm:bernstein-von-mises-mle} and \ref{thm:partial-post-m-est}.
Hence, our algorithm will automatically approach the most efficient
summary statistics in a nonparametric way. 
\end{example}

\section{\label{sec:Discussion}Discussion}

In this paper, we explore ABC both from  theoretical and computational
points of view. The theory part architects the foundation of ABC by
linking asymptotic properties of statistics to that of the partial
posterior. The application part innovates the algorithm by virtue
of bridging selection of summary statistics and SDR. However, although
the theory in \cite{li1992principal} is very powerful and may be
used as a theoretical guide for our algorithm, it heavily depends on
the relation (\ref{eq:sdr}) holding rigorously. We do not know whether
the result from the principal support vector machine would be defunct
if (\ref{eq:sdr}) is only valid in $\varepsilon$-sufficient way.
Moreover, bringing in dimension reduction regression settings perhaps
moderates the usage when there are multiple parameters of interest,
and may need advance techniques such as envelope models of \cite{su2011partial,su2012inner}.

\appendix

\section{Proof of Theorem 1}\label{sec:proof-thm-1}
\begin{proof}
\begin{eqnarray*}
 &  & \prob\left\{ \left(n\hat{I}\right)^{1/2}\left(\theta-\hat{\theta}\right)\le t\mid\hat{\theta}\in O\left(\theta_{0},\varepsilon\right)\right\} \\
 & = & \frac{\prob\left\{ \left(n\hat{I}\right)^{1/2}\left(\theta-\hat{\theta}\right)\le t,\hat{\theta}\in O\left(\theta_{0},\varepsilon\right)\right\} }{\prob\left\{ \hat{\theta}\in O\left(\theta_{0},\varepsilon\right)\right\} }\\
 & = & \frac{\int I_{\left\{ \hat{\theta}\in O\left(\theta_{0},\varepsilon\right)\right\} }I_{\left\{ \left(n\hat{I}\right)^{1/2}\left(\theta-\hat{\theta}\right)\le t\right\} }\prod_{i=1}^{n}f\left(X_{i}\mid\theta\right)\pi\left(\theta\right)\diff X_{i}\diff\theta}{\int I_{\left\{ \hat{\theta}\in O\left(\theta_{0},\varepsilon\right)\right\} }\prod_{i=1}^{n}f\left(X_{i}\mid\theta\right)\pi\left(\theta\right)\diff X_{i}\diff\theta}.
\end{eqnarray*}
Let $\prob^{\infty}\left(\theta\right)$ be the probability measure
on infinite independent and identically distributed sequence $X_{1},\ldots,X_{n},\ldots,$.
Then 
\begin{eqnarray}
 &  & \prob\left\{ \left(n\hat{I}\right)^{1/2}\left(\theta-\hat{\theta}\right)\le t\mid\hat{\theta}\in O\left(\theta_{0},\varepsilon\right)\right\} \nonumber \\
 & = & E_{\pi\left(\theta\right)}\Bigg(E_{\prob^{\infty}\left(\theta_{0}\right)}\Bigg[I_{\left\{ \hat{\theta}\in O\left(\theta_{0},\varepsilon\right)\right\} }I_{\left\{ \left(n\hat{I}\right)^{1/2}\left(\theta-\hat{\theta}\right)\le t\right\} }\nonumber \\
 &  & \exp\left\{ \sum_{i=1}^{n}\log\pi\left(X_{i}\mid\theta\right)-\log\pi\left(X_{i}\mid\theta_{0}\right)\right\} \Bigg]\Bigg)\nonumber \\
 &  & \Bigg/E_{\pi\left(\theta\right)}\Bigg(E_{\prob^{\infty}\left(\theta_{0}\right)}\Bigg[I_{\left\{ \hat{\theta}\in O\left(\theta_{0},\varepsilon\right)\right\} }\nonumber \\
 &  & \exp\left\{ \sum_{i=1}^{n}\log\pi\left(X_{i}\mid\theta\right)-\log\pi\left(X_{i}\mid\theta_{0}\right)\right\} \Bigg]\Bigg)\label{eq:post-change-measure}
\end{eqnarray}
By strong consistency of the maximum likelihood estimator, for any
$\varepsilon>0$, there is an $N>0$, such that for any $n>N$, $\prob\left\{ \hat{\theta}\in O\left(\theta_{0},\varepsilon\right)\mid\theta_{0}\right\} =1$.
Thus, we can drop the indicator $I_{\left\{ \hat{\theta}\in O\left(\theta_{0},\varepsilon\right)\right\} }$
without changing the value in (\ref{eq:post-change-measure}). We
can change the order of integration. So the numerator of (\ref{eq:post-change-measure})
is 
\begin{eqnarray*}
 &  & E_{\prob^{\infty}\left(\theta_{0}\right)}\left(\left\{ \prod_{i=1}^{n}f\left(X_{i}\mid\theta_{0}\right)\right\} ^{-1}E_{\pi\left(\theta\right)}\left[I_{\left\{ \left(n\hat{I}\right)^{1/2}\left(\theta-\hat{\theta}\right)\le t\right\} }\prod_{i=1}^{n}f\left(X_{i}\mid\theta\right)\right]\right)\\
 & = & E_{\prob^{\infty}\left(\theta_{0}\right)}\left[\left\{ \prod_{i=1}^{n}f\left(X_{i}\mid\theta_{0}\right)\right\} ^{-1}\int I_{\left\{ \left(n\hat{I}\right)^{1/2}\left(\theta-\hat{\theta}\right)\le t\right\} }\prod_{i=1}^{n}f\left(X_{i}\mid\theta\right)\pi\left(\theta\right)\diff\theta\right]\\
 & = & E_{\prob^{\infty}\left(\theta_{0}\right)}\Bigg[\left\{ \prod_{i=1}^{n}f\left(X_{i}\mid\theta_{0}\right)\right\} ^{-1}\prob\left\{ \left(n\hat{I}\right)^{1/2}\left(\theta-\hat{\theta}\right)\le t\mid X_{1},\ldots,X_{n}\right\} \\
 &  & E_{\pi\left(\theta\right)}\left\{ \prod_{i=1}^{n}f\left(X_{i}\mid\theta\right)\right\} \Bigg].
\end{eqnarray*}
By Bernstein--von Mises theorem, 
\[
\lim_{n\rightarrow\infty}\prob\left\{ \left(n\hat{I}\right)^{1/2}\left(\theta-\hat{\theta}\right)\le t\mid X_{1},\ldots,X_{n}\right\} =\Phi\left(t\right),\ascv\prob^{\infty}\left(\theta_{0}\right).
\]
Hence, the result holds.
\end{proof}
To prove a similar result about conditioning on posterior mean, we
go through similar steps as in the proof of Theorem 1.
The only needed change is to prove 
\[
\lim_{n\rightarrow\infty}\prob\left[\left(n\hat{I}\right)^{1/2}\left\{ \theta-E\left(\theta\mid X_{1},\ldots,X_{n}\right)\right\} \le t\mid X_{1},\ldots,X_{n}\right]=\Phi\left(t\right),\ascv\prob^{\infty}\left(\theta_{0}\right).
\]
We know from \cite{ghosh2011moment} that with probability 1, (2.2)
holds. By conditioning on $X_{1},\ldots,X_{n}$, both posterior mean
and maximum likelihood estimator are fixed numbers and 
\[
\lim_{n\rightarrow\infty}n^{1/2}\left\{ E\left(\theta\mid X_{1},\ldots,X_{n}\right)-\hat{\theta}\right\} =0.\ascv.
\]
 Hence, {{} }\textcolor{black}{if we assume the CDF of
the full posterior is continuous and asymptotically normal, then 
\begin{eqnarray*}
 &  & \left|\prob\left[\left(n\hat{I}\right)^{1/2}\left\{ \theta-E\left(\theta\mid X_{1},\ldots,X_{n}\right)\right\} \le t\mid X_{1},\ldots,X_{n}\right]-\Phi\left(t\right)\right|\\
 & \le & \Bigg|\prob\left[\left(n\hat{I}\right)^{1/2}\left(\theta-\hat{\theta}\right)+\hat{I}^{1/2}n^{1/2}\left\{ E\left(\theta\mid X_{1},\ldots,X_{n}\right)-\hat{\theta}\right\} \le t\mid X_{1},\ldots,X_{n}\right]\\
 &  & -\prob\left\{ \left(n\hat{I}\right)^{1/2}\left(\theta-\hat{\theta}\right)\le t\mid X_{1},\ldots,X_{n}\right\} \Bigg|\\
 &  & +\left|\prob\left\{ \left(n\hat{I}\right)^{1/2}\left(\theta-\hat{\theta}\right)\le t\mid X_{1},\ldots,X_{n}\right\} -\Phi\left(t\right)\right|\\
 & \rightarrow & 0,\mathrm{\: as\:}\left(n\rightarrow\infty\right).
\end{eqnarray*}
}

\section{Derivation of Examples}

\subsection{\label{sub:Derivation-of-Example-1}Derivation of Example 1}

First, we get the distribution of $\mu$. 
\[
\mu^{r_{2}}\exp\left(-\mu\frac{r_{2}}{\hat{\mu}}\right)\left|-\frac{r_{2}}{\left(\hat{\mu}\right)^{2}}\right|\propto r_{2}\mu^{r_{2}}\exp\left(-\frac{\mu}{\hat{\mu}}r_{2}\right).
\]
Now we sum over $r_{2}$. By definition $X_{0}+r_{1}-r_{2}\ge0$.
Hence the distribution of $\hat{\mu}$ is proportional to 
\[
\sum_{r_{2}=0}^{X_{0}+r_{1}}r_{2}\left\{ \mu\exp\left(-\frac{\mu}{\hat{\mu}}\right)\right\} ^{r_{2}}.
\]
Let $U=\mu\exp\left(-\mu/\hat{\mu}\right)$ and $R=X_{0}+r_{1}=X_{0}+\hat{\lambda}T$.
Let 
\begin{eqnarray*}
L & = & \sum_{r_{2}=0}^{R}r_{2}U^{r_{2}}=U\frac{\diff}{\diff U}\sum_{r_{2}=0}^{R}U^{r_{2}}=U\frac{\diff}{\diff U}\left(\frac{1-U^{R+1}}{1-U}\right)\\
 & = & U\left(1-U\right)^{-2}\left\{ 1-\left(R+1\right)U^{R}+RU^{R+1}\right\} .
\end{eqnarray*}
For fixed $t$, consider $\mu=\hat{\mu}+t/T^{1/2}$. 
\begin{eqnarray*}
\log U & = & -\frac{\mu}{\hat{\mu}}+\log\mu=-1-\frac{t}{T^{1/2}\hat{\mu}}+\log\hat{\mu}+\log\left(1+\frac{t}{T^{1/2}\hat{\mu}}\right)\\
 & = & \log\hat{\mu}-1-\frac{t^{2}}{2\left(\hat{\mu}\right)^{2}T}+o\left(T^{-1}\right).
\end{eqnarray*}
\textcolor{black}{Hence 
\[
\lim_{T\rightarrow\infty}U=\frac{\mu_{0}}{e},\ascv..
\]
Next, 
\[
R\log U=\left(X_{0}+\hat{\lambda}T\right)\left(\log\hat{\mu}-1\right)-\frac{\hat{\lambda}t^{2}}{2\left(\hat{\mu}\right)^{2}}+o\left(1\right).
\]
\[
L=\frac{UR}{\left(1-U\right)^{2}}\left[\frac{1}{R}-\frac{R+1}{R}\exp\left(R\log U\right)+U\exp\left(R\log U\right)\right]
\]
Note that the density of $t=T^{1/2}\left(\mu-\hat{\mu}\right)$ is
only proportional to $L$, hence only the terms containing $t$ will
affect the limit distribution, other terms can be omitted. Also recall
that, 
\[
\lim_{T\rightarrow\infty}\frac{1}{R}=\lim_{T\rightarrow\infty}\frac{1}{X_{0}+\hat{\lambda}T}=0.
\]
Thus 
\begin{eqnarray*}
\lim_{T\rightarrow\infty}L & \propto & \lim_{T\rightarrow\infty}\left\{ \frac{1}{R}-\frac{R+1}{R}\exp\left(R\log U\right)+U\exp\left(R\log U\right)\right\} \\
 & = & \lim_{T\rightarrow\infty}\left(U-1-\frac{1}{R}\right)\exp\left\{ \left(X_{0}+\hat{\lambda}T\right)\left(\log\hat{\mu}-1\right)-\frac{\hat{\lambda}t^{2}}{2\left(\hat{\mu}\right)^{2}}+o\left(1\right)\right\} \\
 & = & \lim_{T\rightarrow\infty}\left(\frac{\mu_{0}}{e}-1\right)\exp\left\{ \left(X_{0}+\hat{\lambda}T\right)\left(\log\hat{\mu}-1\right)\right\} \exp\left\{ -\frac{\hat{\lambda}t^{2}}{2\left(\hat{\mu}\right)^{2}}+o\left(1\right)\right\} \\
 & \propto & \exp\left\{ -\frac{\hat{\lambda}t^{2}}{2\left(\hat{\mu}\right)^{2}}\right\} .
\end{eqnarray*}
}

\subsection{\label{sub:Derivation-of-Example-2}Derivation of Example 2}

We know $\overline{X}\sim\mathrm{Gamma}\left(n\alpha,\beta/n\right),$
so $\tilde{\alpha}\sim\mathrm{Gamma}\left(n\alpha,n^{-1}\right).$
\begin{eqnarray*}
\pi\left(\alpha\mid\tilde{\alpha}\right) & \propto & \pi\left(\tilde{\alpha}\mid\alpha\right)\pi\left(\alpha\right)=\frac{1}{\Gamma\left(n\alpha\right)n^{-n\alpha}}\left(\tilde{\alpha}\right)^{n\alpha-1}\exp\left(-n\tilde{\alpha}\right)\exp\left(-\lambda\alpha\right)\\
 & \propto & \frac{\left\{ n\tilde{\alpha}\exp\left(-\lambda/n\right)\right\} ^{n\alpha}}{\Gamma\left(n\alpha\right)}.
\end{eqnarray*}
Next we will show $\pi\left(n^{1/2}\left(\alpha-\tilde{\alpha}\right)/b\mid\tilde{\alpha}\right)\rightarrow N\left(0,1\right)\ascv,$
for some suitable $b$. The PDF of $t$ is proportional to 
\[
\frac{\left\{ n\tilde{\alpha}\exp\left(-\lambda/n\right)\right\} ^{n\left(bt/n^{1/2}+\tilde{\alpha}\right)}}{\Gamma\left\{ n\left(bt/n^{1/2}+\tilde{\alpha}\right)\right\} }\frac{b}{n^{1/2}}.
\]
Take logarithm, and drop all the terms not related to $t$, since
those terms can be divided from both numerator and denominator, 
\begin{equation}
n^{1/2}bt\log\left(n\tilde{\alpha}\right)-\lambda b\frac{t}{n^{1/2}}-\log\Gamma\left\{ n\tilde{\alpha}\left(1+\frac{bt}{n^{1/2}\tilde{\alpha}}\right)\right\} .\label{eq:log-density-drop-not-have-t}
\end{equation}
Using Stirling formula to approximate gamma function, 
\begin{eqnarray*}
 &  & \log\Gamma\left\{ n\tilde{\alpha}\left(1+\frac{bt}{n^{1/2}\tilde{\alpha}}\right)\right\} \\
 & \approx & \left\{ n\tilde{\alpha}\left(1+\frac{bt}{n^{1/2}\tilde{\alpha}}\right)-\frac{1}{2}\right\} \log\left\{ n\tilde{\alpha}\left(1+\frac{bt}{n^{1/2}\tilde{\alpha}}\right)\right\} -n\tilde{\alpha}\left(1+\frac{bt}{n^{1/2}\tilde{\alpha}}\right)+\frac{1}{2}\log2\pi.
\end{eqnarray*}
So (\ref{eq:log-density-drop-not-have-t}) can be written as 
\begin{eqnarray}
 &  & n^{1/2}bt\log\left(n\tilde{\alpha}\right)-\lambda b\frac{t}{n^{1/2}}-n^{1/2}bt\log\left(1+\frac{bt}{n^{1/2}\tilde{\alpha}}\right)\nonumber \\
 &  & -n^{1/2}bt\log\left(n\tilde{\alpha}\right)-\left(n\tilde{\alpha}-\frac{1}{2}\right)\log\left(1+\frac{bt}{n^{1/2}\tilde{\alpha}}\right)+n^{1/2}bt\nonumber \\
 & = & -n^{1/2}bt\log\left(1+\frac{bt}{n^{1/2}\tilde{\alpha}}\right)-\left(n\tilde{\alpha}-\frac{1}{2}\right)\log\left(1+\frac{bt}{n^{1/2}\tilde{\alpha}}\right)\label{eq:log-density-after-stirling}\\
 &  & +n^{1/2}bt-\lambda b\frac{t}{n^{1/2}}.\nonumber 
\end{eqnarray}
Now we can apply Taylor expansion for term $\log\left\{ 1+bt/\left(n^{1/2}\tilde{\alpha}\right)\right\} $,
\[
\log\left(1+\frac{bt}{n^{1/2}\tilde{\alpha}}\right)=\frac{bt}{n^{1/2}\tilde{\alpha}}-\frac{b^{2}t^{2}}{2n\tilde{\alpha}^{2}}+\frac{b^{3}t^{3}}{3n^{3/2}\tilde{\alpha}^{3}}+o\left(\frac{t^{3}}{n^{3/2}}\right).
\]
Substituting the expansion into (\ref{eq:log-density-after-stirling}),
\begin{eqnarray*}
 &  & -n^{1/2}bt\log\left(1+\frac{bt}{n^{1/2}\tilde{\alpha}}\right)-\left(n\tilde{\alpha}-\frac{1}{2}\right)\log\left(1+\frac{bt}{n^{1/2}\tilde{\alpha}}\right)+n^{1/2}bt-\lambda b\frac{t}{n^{1/2}}\\
 & = & -n^{1/2}bt\left\{ \frac{bt}{n^{1/2}\tilde{\alpha}}-\frac{b^{2}t^{2}}{2n\tilde{\alpha}^{2}}+\frac{b^{3}t^{3}}{3n^{3/2}\tilde{\alpha}^{3}}+o\left(\frac{t^{3}}{n^{3/2}}\right)\right\} \\
 &  & -\left(n\tilde{\alpha}-\frac{1}{2}\right)\left\{ \frac{bt}{n^{1/2}\tilde{\alpha}}-\frac{b^{2}t^{2}}{2n\tilde{\alpha}^{2}}+\frac{b^{3}t^{3}}{3n^{3/2}\tilde{\alpha}^{3}}+o\left(\frac{t^{3}}{n^{3/2}}\right)\right\} +n^{1/2}bt-\lambda b\frac{t}{n^{1/2}}\\
 & = & -\frac{b^{2}t^{2}}{\tilde{\alpha}}+\frac{b^{3}t^{3}}{2n^{1/2}\tilde{\alpha}^{2}}-o\left(\frac{t^{3}}{n^{1/2}}\right)-n^{1/2}bt+\frac{b^{2}t^{2}}{2\tilde{\alpha}}-\frac{b^{3}t^{3}}{2n^{1/2}\alpha^{2}}-o\left(\frac{t^{3}}{n^{1/2}}\right)\\
 &  & +\frac{bt}{2n^{1/2}\tilde{\alpha}}-\frac{b^{2}t^{2}}{4n\tilde{\alpha}^{2}}+o\left(\frac{t^{2}}{n}\right)+n^{1/2}bt-\lambda b\frac{t}{n^{1/2}}\\
 & \approx & -\frac{b^{2}t^{2}}{2\tilde{\alpha}}.
\end{eqnarray*}
If we set $b=\tilde{\alpha}^{1/2}$, then the rescaled partial posterior
convergence to standard normal.

\subsection{\label{sub:Derivation-of-Example-3}Derivation of Example 3}

\textcolor{black}{First we prove lemma 1.}
\begin{proof}
\textcolor{black}{Assume $Y$ has a probability density function $f\left(y\right)$.
Let $y=u\left(x,\theta\right)$ be the solution of equation $x=h\left(y,\theta\right)$.
Then $h\left(Y,\theta\right)$ has a probability density function
\[
f\left(u\left(x,\theta\right)\right)\left|\frac{\partial u\left(x,\theta\right)}{\partial x}\right|.
\]
Then the posterior distribution under the uniform prior is proportional
to 
\[
f\left(u\left(x,\theta\right)\right)\left|\frac{\partial u\left(x,\theta\right)}{\partial x}\right|.
\]
Now we find the probability density function of $g\left(Y,X\right)$.
By assumptions, we know $y=u\left(x,\theta\right)$ is also the solution
of $\theta=g\left(y,x\right)$. Hence the probability density function
of $g\left(Y,X\right)$ is also 
\[
f\left(u\left(x,\theta\right)\right)\left|\frac{\partial u\left(x,\theta\right)}{\partial x}\right|.
\]
}
\end{proof}

We know if $X,Y$ are independent exponential random variables with
mean $\lambda$, then $X-Y$ has a double exponential distribution
with $\mu=0$ and the same $\lambda$. So we know our sample $Z$
has the same distribution as $X-Y+\mu$. So the sample mean $\overline{Z}$
has the same distribution as $\overline{X}-\overline{Y}+\mu$. It
is easy to check $\overline{X}$ and $\overline{Y}$ have gamma distribution
with location parameter $n$ and scale parameter $n^{-1}\lambda$.
Hence the posterior distribution of $\mu$on $\overline{Z}$ under
the uniform prior has the same distribution as $\overline{Z}-\left(\overline{X}-\overline{Y}\right)$.
Hence the posterior distribution $n^{1/2}\left(\mu-\overline{Z}\right)$
has the same distribution as $-n^{1/2}\left(\overline{X}-\overline{Y}\right)$
We know the characteristic function of $n^{1/2}$$\overline{X}$ is
\[
\left\{ 1-\frac{\lambda}{n}i\left(n^{1/2}t\right)\right\} ^{-n},
\]
So the characteristic function of $-n^{1/2}\left(\overline{X}-\overline{Y}\right)$
is 
\[
\left\{ 1-\frac{\lambda}{n}i\left(n^{1/2}t\right)\right\} ^{-n}\left\{ 1-\frac{\lambda}{n}i\left(-n^{1/2}t\right)\right\} ^{-n}=\left(1+\frac{\lambda^{2}t^{2}}{n}\right)^{-n}\rightarrow\exp\left(-\lambda^{2}t^{2}\right).
\]
Hence $n^{/1}\left(\mu-\overline{Z}\right)$ has an asymptotic normal
distribution with zero mean and variance $2\lambda^{2}$.

\section{Proof of Theorem 2}\label{sec:Proof-of-Theorem-2}
\begin{lemma}
\label{lem:taylor-expansion-in-dist}Under Assumptions 1,
2 and 3,
for any $\varepsilon$, $\delta_{1}$ and $\delta_{2}$, there exists
an $N$, such that for any $n\ge N$, 
\[
\prob_{\theta_{0}}^{\infty}\left[\omega:\prob_{\omega}^{n}\left\{ n^{1/2}\left|G\left(\theta,\tilde{\theta}\right)-G_{1}\left(\hat{\theta},\tilde{\theta}\right)\left(\theta-\tilde{\theta}\right)\right|\le2\varepsilon\right\} \ge1-\delta_{1}\right]\ge1-\delta_{2}.
\]
\end{lemma}
\begin{proof}
By Taylor expansion and Assumption 1,
\begin{equation}
\left|G\left(\theta,\tilde{\theta}\right)-G\left(\hat{\theta},\tilde{\theta}\right)-G_{1}\left(\hat{\theta},\tilde{\theta}\right)\left(\theta-\hat{\theta}\right)\right|\le L\left(\theta-\hat{\theta}\right)^{2},\label{eq:second-lip-1}
\end{equation}
and 
\begin{equation}
\left|G\left(\tilde{\theta},\tilde{\theta}\right)-G\left(\hat{\theta},\tilde{\theta}\right)-G_{1}\left(\hat{\theta},\tilde{\theta}\right)\left(\tilde{\theta}-\hat{\theta}\right)\right|\le L\left(\tilde{\theta}-\hat{\theta}\right)^{2}.\label{eq:second-lip-2}
\end{equation}
By posterior consistency, there exists a $\Omega_{1}\subset\Omega$,
$\prob_{\theta_{0}}^{\infty}\left(\Omega_{1}\right)=1$, such that
for any $\omega\in\Omega_{1}$, conditioned on $X_{1}\left(\omega\right),\ldots,X_{n}\left(\omega\right)$, $\theta$
converges in probability to $\theta_{0}$. Hence conditioned on $X_{1}\left(\omega\right),\ldots,X_{n}\left(\omega\right)$,  $ \left(\theta-\theta_{0}\right) =o_{\prob_{\omega}^{n}}\left(1\right)$.
By Bernstein--von Mises, there exists a $\Omega_{2}\subset\Omega$,
$\prob_{\theta_{0}}^{\infty}\left(\Omega_{2}\right)=1$, such that
for any $\omega\in\Omega_{2}$, conditioned on $X_{1}\left(\omega\right),\ldots,X_{n}\left(\omega\right)$ , $\left(n\hat{I}\right)^{1/2}\left(\theta-\hat{\theta}\right)$
converges in distribution to a standard normal random variable. Hence
conditioned on $X_{1}\left(\omega\right),\ldots,X_{n}\left(\omega\right)$, $ \left(n\hat{I}\right)^{1/2}\left(\theta-\hat{\theta}\right) =O_{\prob_{\omega}^{n}}\left(1\right).$
For any $\omega\in\Omega_{1}\cap\Omega_{2}$, 
\begin{eqnarray*}
 &  & \left\{ n^{1/2}L\left(\theta-\hat{\theta}\right)^{2}\mid X_{1}\left(\omega\right),\ldots,X_{n}\left(\omega\right)\right\} \\
 & = & L\left\{ n^{1/2}\left(\theta-\hat{\theta}\right)\times\left(\theta-\hat{\theta}\right)\mid X_{1}\left(\omega\right),\ldots,X_{n}\left(\omega\right)\right\} \\
 & = & LO_{\prob_{\omega}^{n}}\left(1\right)\times o_{\prob_{\omega}^{n}}\left(1\right)=o_{\prob_{\omega}^{n}}\left(1\right),
\end{eqnarray*}
which means for any $\varepsilon$ and $\delta_{1}$, there exists
an $N_{1}$ such that 
\[
\prob_{\omega}^{n}\left[n^{1/2}L\left\{ \theta-\hat{\theta}\left(\omega\right)\right\} ^{2}\le\varepsilon\mid X_{1}\left(\omega\right),\ldots,X_{n}\left(\omega\right)\right]\ge1-\delta_{1}.
\]
By Assumptions 2, $\tilde{\theta}-\theta_{0}=o_{\prob_{\theta_{0}}^{\infty}}\left(1\right)$,
$n^{1/2}\left(\tilde{\theta}-\theta_{0}\right)=O_{\prob_{\theta_{0}}^{\infty}}\left(1\right)$,
$\hat{\theta}-\theta_{0}=o_{\prob_{\theta_{0}}^{\infty}}\left(1\right)$
and $n^{1/2}\left(\hat{\theta}-\theta_{0}\right)=O_{\prob_{\theta_{0}}^{\infty}}\left(1\right).$Then
\[
n^{1/2}L\left(\tilde{\theta}-\hat{\theta}\right)^{2}\le2L\left\{ n^{1/2}\left(\tilde{\theta}-\theta_{0}\right)^{2}+n^{1/2}\left(\hat{\theta}-\theta_{0}\right)^{2}\right\} =o_{\prob_{\theta_{0}}^{\infty}}\left(1\right),
\]
which means for any $\varepsilon$ and $\delta_{2}$, there exists
an $N_{2}$, such that for any $n\ge N_{2}$, 
\[
\prob_{\theta_{0}}^{\infty}\left[\omega:n^{1/2}L\left\{ \tilde{\theta}\left(\omega\right)-\hat{\theta}\left(\omega\right)\right\} ^{2}\le\varepsilon\right]\ge1-\delta_{2}.
\]
Let $\Omega_{\varepsilon}=\left\{ \omega:n^{1/2}L\left\{ \tilde{\theta}\left(\omega\right)-\hat{\theta}\left(\omega\right)\right\} ^{2}\le\varepsilon\right\} $.
For any $\omega\in\Omega_{1}\cap\Omega_{2}\cap\Omega_{\varepsilon}$,
\begin{eqnarray*}
 &  & n^{1/2}\left|G\left(\theta,\tilde{\theta}\right)-G_{1}\left(\hat{\theta},\tilde{\theta}\right)\left(\theta-\tilde{\theta}\right)\right|\\
 & \le & n^{1/2}\left|G\left(\theta,\tilde{\theta}\right)-G\left(\hat{\theta},\tilde{\theta}\right)-G_{1}\left(\hat{\theta},\tilde{\theta}\right)\left(\theta-\hat{\theta}\right)\right|\\
 &  & +n^{1/2}\left|G\left(\tilde{\theta},\tilde{\theta}\right)-G\left(\hat{\theta},\tilde{\theta}\right)-G_{1}\left(\hat{\theta},\tilde{\theta}\right)\left(\tilde{\theta}-\hat{\theta}\right)\right|\\
 & \le & n^{1/2}L\left(\theta-\hat{\theta}\right)^{2}+n^{1/2}L\left(\tilde{\theta}-\hat{\theta}\right)^{2}\le2\varepsilon,
\end{eqnarray*}
with probability $1-\delta_{1}$ $\left(\prob_{\omega}^{n}\right)$.
Also recall $\prob_{\theta_{0}}^{\infty}\left(\Omega_{1}\cap\Omega_{2}\cap\Omega_{\varepsilon}\right)\ge1-\delta_{2}$.
Hence for $\omega\in\Omega_{1}\cap\Omega_{2}\cap\Omega_{\varepsilon}$
and $n\ge\max\left\{ N_{1},N_{2}\right\} $, 
\[
\prob_{\theta_{0}}^{\infty}\left[\omega:P_{\omega}^{n}\left\{ n^{1/2}\left|G\left(\theta,\tilde{\theta}\right)-G_{1}\left(\hat{\theta},\tilde{\theta}\right)\left(\theta-\tilde{\theta}\right)\right|\le2\varepsilon\right\} \ge1-\delta_{1}\right]\ge1-\delta_{2}.
\]
\end{proof}
\begin{remark}
This result is weaker than the settings in posterior consistency and
Bernstein--von Mises theorem. In posterior consistency, $\left(\theta\mid X_{1},\ldots,X_{n}\right)$
converges in probability $\left(\prob_{\omega}^{n}\right)$ almost
surely in $\prob_{\theta_{0}}^{\infty}$. A similar comment applies to the Bernstein--von Mises theorem. However, in this lemma, the posterior
distribution converges with a large probability. 
\begin{remark}
There is slight difference between the rescaled posterior random variable
in this lemma and in 
\begin{eqnarray}
 &  & n^{1/2}\Bigg[\int_{\mathbb{R}}g\left(x,\tilde{\theta}\right)\pi\left(x\mid\theta\right)\diff x-\int_{\mathbb{R}}g\left(x,\tilde{\theta}\right)\pi\left(x\mid\tilde{\theta}\right)\diff x\label{eq:strong-vanish-error}\\
 &  & -\left\{ \left.\frac{\diff}{\diff\theta}\int_{\mathbb{R}}g\left(x,\tilde{\theta}\right)\pi\left(x\mid\theta\right)\diff x\right|_{\theta=\tilde{\theta}}\right\} \left(\theta-\tilde{\theta}\right)\Bigg]\rightarrow0,\ascv.\nonumber 
\end{eqnarray}
 The first order differential term is $G_{1}\left(\tilde{\theta},\tilde{\theta}\right)$.
However, since $G_{1}\left(\tilde{\theta},\tilde{\theta}\right)\rightarrow G_{1}\left(\theta_{0},\theta_{0}\right)$
and $G_{1}\left(\hat{\theta},\tilde{\theta}\right)\rightarrow G_{1}\left(\theta_{0},\theta_{0}\right)$
almost surely in $P_{\theta_{0}}^{\infty}$ and $n^{1/2}$ term is
absorbed by $\left(\theta-\tilde{\theta}\right)$, we have proved
a weak version of (\ref{eq:strong-vanish-error}). 
\end{remark}
\end{remark}
\begin{proof}
Under the Assumptions 1,
2 and 3,
we have the result from Lemma \ref{lem:taylor-expansion-in-dist},
\begin{equation}
\prob_{\theta_{0}}^{\infty}\left[\omega:\prob_{\omega}^{n}\left\{ n^{1/2}\left|G\left(\theta,\tilde{\theta}\right)-G_{1}\left(\hat{\theta},\tilde{\theta}\right)\left(\theta-\tilde{\theta}\right)\right|\le2\varepsilon\right\} \ge1-\delta_{1}\right]\ge1-\delta_{2}.\label{eq:conv-in-dist-reminder}
\end{equation}
Let $\Omega_{1}=\left\{ \omega:\prob_{\omega}^{n}\left\{ n^{1/2}\left|G\left(\theta,\tilde{\theta}\right)-G_{1}\left(\hat{\theta},\tilde{\theta}\right)\left(\theta-\tilde{\theta}\right)\right|\le2\varepsilon\right\} \ge1-\delta_{1}\right\} $, and 
\[ 
C_{n}=E_{\pi\left(\theta\right)}\left(E_{\prob^{\infty}\left(\theta_{0}\right)}\left[I_{\left\{ \tilde{\theta}\in O\left(\theta_{0},\varepsilon\right)\right\} }\exp\left\{ \sum_{i=1}^{n}\log f\left(X_{i}\mid\theta\right)-\log f\left(X_{i}\mid\theta_{0}\right)\right\} \right]\right).\]

Now by the same technique used in the proof of Theorem 1,
we have for sufficiently large $N$, 
\begin{eqnarray*}
 &  & \left|\prob\left\{ \left.\frac{n^{1/2}\left(\theta-\tilde{\theta}\right)}{\tilde{V}^{1/2}}\le t\right|\tilde{\theta}\in O\left(\theta_{0},\varepsilon\right)\right\} -\Phi\left(t\right)\right|\\
 & = & \Bigg|E_{\pi\left(\theta\right)}\left(E_{\prob^{\infty}\left(\theta_{0}\right)}\left[I_{\left\{ \tilde{\theta}\in O\left(\theta_{0},\varepsilon\right)\right\} }I_{\left\{ n^{1/2}\tilde{V}^{-1/2}\left(\theta-\tilde{\theta}\right)\le t\right\} }\prod_{i=1}^{n}f\left(X_{i}\mid\theta\right)/f\left(X_{i}\mid\theta_{0}\right)\right]\right)\\
 &  & \Bigg/E_{\pi\left(\theta\right)}\left(E_{\prob^{\infty}\left(\theta_{0}\right)}\left[I_{\left\{ \tilde{\theta}\in O\left(\theta_{0},\varepsilon\right)\right\} }\prod_{i=1}^{n}f\left(X_{i}\mid\theta\right)/f\left(X_{i}\mid\theta_{0}\right)\right]\right)-\Phi\left(t\right)\Bigg|\\
 & \le & C_{n}^{-1}E_{\pi\left(\theta\right)}\Bigg(E_{\prob^{\infty}\left(\theta_{0}\right)}\Bigg[I_{\Omega_{1}}\left|\prob\left\{ n^{1/2}\tilde{V}^{-11/2}\left(\theta-\tilde{\theta}\right)\le t\mid X_{1},\ldots,X_{n}\right\} -\Phi\left(t\right)\right|\\
 &  & \prod_{i=1}^{n}f\left(X_{i}\mid\theta\right)/f\left(X_{i}\mid\theta_{0}\right)\Bigg]\Bigg)\\
 &  & +C_{n}^{-1}E_{\pi\left(\theta\right)}\Bigg(E_{\prob^{\infty}\left(\theta_{0}\right)}\Bigg[I_{\Omega_{1}^{c}}\left|\prob\left\{ n^{1/2}\tilde{V}^{-11/2}\left(\theta-\tilde{\theta}\right)\le t\mid X_{1},\ldots,X_{n}\right\} -\Phi\left(t\right)\right|\\
 &  & \prod_{i=1}^{n}f\left(X_{i}\mid\theta\right)/f\left(X_{i}\mid\theta_{0}\right)\Bigg]\Bigg).
\end{eqnarray*}

First considering the samples within $\Omega_{1}$, by (\ref{eq:conv-in-dist-reminder}),
$V_{0}^{-1/2}\left[n^{1/2}\left\{ G_{1}\left(\tilde{\theta},\tilde{\theta}\right)\left(\theta-\tilde{\theta}\right)-G\left(\theta,\tilde{\theta}\right)\right\} \right]$
converges in probability to 0. By Theorem 2.1 in \cite{rivoirard2012bernstein},
we have, 
\begin{eqnarray}
 &  & \lim_{n\rightarrow\infty}\sup_{t\in\mathbb{R}}\Bigg|\prob\Bigg[n^{1/2}V_{0}^{-1/2}\left\{ G\left(\theta,\tilde{\theta}\right)-\frac{1}{n}\sum_{i=1}^{n}g\left(X_{i},\tilde{\theta}\right)\right\} \label{eq:conv-in-dist-est-eq}\\
 &  & \le t\mid X_{1},\ldots,X_{n}\Bigg]-\Phi\left(t\right)\Bigg|=0,\ascv\theta_{0},\nonumber 
\end{eqnarray}
where 
\[
V_{0}=\int_{\mathbb{R}}\left\{ g\left(x,\tilde{\theta}\right)-\int_{\mathbb{R}}g\left(y,\tilde{\theta}\right)\pi\left(y\mid\theta_{0}\right)\diff y\right\} ^{2}\pi\left(x\mid\theta_{0}\right)\diff x.
\]
Hence, $n^{1/2}V_{0}^{-1/2}\left\{ G\left(\theta,\tilde{\theta}\right)-n^{-1}\sum_{i=1}^{n}g\left(X_{i},\tilde{\theta}\right)\right\} $
converges in distribution to standard normal distribution. By the
definition of an $M$-estimator, $n^{-1}\sum_{i=1}^{n}g\left(X_{i},\tilde{\theta}\right)=0$.
Assume that for every $\theta$, $\int_{\mathbb{R}}g\left(x,t\right)\pi\left(x\mid\theta\right)\diff x=0$
has only one solution $t=\theta$, then 
\[
G\left(\tilde{\theta},\tilde{\theta}\right)=\int_{\mathbb{R}}g\left(x,\tilde{\theta}\right)\pi\left(x\mid\tilde{\theta}\right)\diff x=0.
\]
Hence, by Slutsky's Theorem,
\begin{eqnarray*}
n^{1/2}\tilde{V}^{-1/2}\left(\theta-\tilde{\theta}\right) & = & n^{1/2}V_{0}^{-1/2}\left\{ G\left(\theta,\tilde{\theta}\right)-\frac{1}{n}\sum_{i=1}^{n}g\left(X_{i},\tilde{\theta}\right)\right\} \\
 &  & +V_{0}^{-1/2}\left[n^{1/2}\left\{ G_{1}\left(\tilde{\theta},\tilde{\theta}\right)\left(\theta-\tilde{\theta}\right)-G\left(\theta,\tilde{\theta}\right)\right\} \right],
\end{eqnarray*}
 converges in distribution to the standard normal distribution. Hence
for large $N$, 
\[
\sup_{t\in\mathbb{R}}\left|\prob\left\{ n^{1/2}\tilde{V}^{-1/2}\left(\theta-\tilde{\theta}\right)\le t\mid X_{1},\ldots,X_{n}\right\} -\Phi\left(t\right)\right|\le\varepsilon,
\]
and 
\begin{eqnarray}
 &  & C_{n}^{-1}E_{\pi\left(\theta\right)}\Bigg(E_{\prob^{\infty}\left(\theta_{0}\right)}\Bigg[I_{\Omega_{1}}\sup_{t\in\mathbb{R}}\left|\prob\left\{ n^{1/2}\tilde{V}^{-1/2}\left(\theta-\tilde{\theta}\right)\le t\mid X_{1},\ldots,X_{n}\right\} -\Phi\left(t\right)\right|\nonumber \\
 &  & \prod_{i=1}^{n}f\left(X_{i}\mid\theta\right)/f\left(X_{i}\mid\theta_{0}\right)\Bigg]\Bigg)\nonumber \\
 & \le & \varepsilon C_{n}^{-1}E_{\pi\left(\theta\right)}\left[E_{\prob^{\infty}\left(\theta_{0}\right)}\left\{ I_{\Omega_{1}}\prod_{i=1}^{n}f\left(X_{i}\mid\theta\right)/f\left(X_{i}\mid\theta_{0}\right)\right\} \right]=\varepsilon.\label{eq:conv-with-omega1}
\end{eqnarray}

For samples outside $\Omega_{1}$. It is trivial that 
\[
\sup_{t\in\mathbb{R}}\left|\prob\left\{ n^{1/2}\tilde{V}^{-1/2}\left(\theta-\tilde{\theta}\right)\le t\mid X_{1},\ldots,X_{n}\right\} -\Phi\left(t\right)\right|\le2.
\]
By Assumption 4 and the strong law of large numbers,
and the property of the Kullback-Leibler information number 
\begin{eqnarray*}
 &  & \prod_{i=1}^{n}f\left(X_{i}\mid\theta\right)/f\left(X_{i}\mid\theta_{0}\right)\\
 & = & \exp\left[n\left\{ \frac{1}{n}\sum_{i=1}^{n}\log f\left(X_{i}\mid\theta\right)-\frac{1}{n}\sum_{i=1}^{n}\log f\left(X_{i}\mid\theta_{0}\right)\right\} \right]\le1.\ascv\left(P_{\theta_{0}}\right)
\end{eqnarray*}
Hence
\begin{eqnarray}
 &  & C_{n}^{-1}E_{\pi\left(\theta\right)}\Bigg(E_{\prob^{\infty}\left(\theta_{0}\right)}\Bigg[I_{\Omega_{1}^{c}}\sup_{t\in\mathbb{R}}\left|\prob\left\{ n^{1/2}\tilde{V}^{-1/2}\left(\theta-\tilde{\theta}\right)\le t\mid X_{1},\ldots,X_{n}\right\} -\Phi\left(t\right)\right|\nonumber \\
 &  & \prod_{i=1}^{n}f\left(X_{i}\mid\theta\right)/f\left(X_{i}\mid\theta_{0}\right)\Bigg]\Bigg)\nonumber \\
 & \le & 2C_{n}^{-1}E_{\pi\left(\theta\right)}\left\{ E_{\prob^{\infty}\left(\theta_{0}\right)}\left(I_{\Omega_{1}^{c}}\right)\right\} =2\prob^{\infty}\left(\Omega_{1}^{c}\mid\theta_{0}\right)=2\delta_{2}.\label{eq:conv-outside-omega1}
\end{eqnarray}

Hence, combining (\ref{eq:conv-with-omega1}) and (\ref{eq:conv-outside-omega1}),
\[
\sup_{t\in\mathbb{R}}\left|\prob\left\{ \left.\frac{n^{1/2}\left(\theta-\tilde{\theta}\right)}{\tilde{V}^{1/2}}\le t\right|\tilde{\theta}\in O\left(\theta_{0},\varepsilon\right)\right\} -\Phi\left(t\right)\right|\le\varepsilon+2\delta_{2},\ascv
\]

\end{proof}

\section{Proof of Theorem 3}\label{sec:Proof-of-Theorem-3}
\begin{proof}
By Theorem 2.1 in \cite{rivoirard2012bernstein}, we have 
\begin{eqnarray}
 &  & \lim_{n\rightarrow\infty}\sup_{t\in\mathbb{R}}\Bigg|\prob\Bigg[n^{1/2}\mathrm{var}_{\theta_{0}}\left\{ a^{T}g\left(X\right)\right\} ^{-1/2}\label{eq:thm2-1-inconsist}\\
 &  & \left\{ \int a^{T}g\left(x\right)f\left(x\mid\theta\right)\diff x-\frac{1}{n}\sum_{i=1}^{n}a^{T}g\left(X_{i}\right)\right\} \nonumber \\
 &  & \le t\mid X_{1},\ldots,X_{n}\Bigg]-\Phi\left(t\right)\Bigg|=0,\ascv.\nonumber 
\end{eqnarray}
By Assumption 6 and Slutsky's theorem,
we know $n^{-1/2}\sum_{i=1}^{n}g\left(X_{i}\right)$ has the same
asymptotic distribution as $S$. However, by central limit theorem,
$n^{-1/2}\sum_{i=1}^{n}g\left(X_{i}\right)$ has an asymptotic normal
distribution with variance matrix as $\mathrm{var}_{\theta_{0}}\left\{ g\left(X\right)\right\} $.
Hence, $\mathrm{var}_{\theta_{0}}\left\{ g\left(X\right)\right\} =\Sigma\left(\theta_{0}\right)=\lim_{n\rightarrow\infty}\tilde{\Sigma},\ascv.$
Hence, we can replace $\mathrm{var}_{\theta_{0}}\left\{ g\left(X\right)\right\} $
in (\ref{eq:thm2-1-inconsist}) by its strong consistent estimator,
and get
\begin{eqnarray*}
 &  & \lim_{n\rightarrow\infty}\sup_{t\in\mathbb{R}}\Bigg|\prob\Bigg[n^{1/2}\left(a^{T}\tilde{\Sigma}a\right)^{-1/2}\left\{ \int a^{T}g\left(x\right)f\left(x\mid\theta\right)\diff x-n^{-1}\sum_{i=1}^{n}a^{T}g\left(X_{i}\right)\right\} \\
 &  & \le t\mid X_{1},\ldots,X_{n}\Bigg]-\Phi\left(t\right)\Bigg|=0,\ascv.
\end{eqnarray*}
By Assumption 6, we can replace
$n^{1/2}\left\{ n^{-1}\sum_{i=1}^{n}a^{T}g\left(X_{i}\right)\right\} $
by $n^{1/2}a^{T}S$, and finally obtain
\begin{eqnarray*}
 &  & \lim_{n\rightarrow\infty}\sup_{t\in\mathbb{R}}\Bigg|\prob\Bigg[n^{1/2}\left(a^{T}\tilde{\Sigma}a\right)^{-1/2}\left\{ \int a^{T}g\left(x\right)f\left(x\mid\theta\right)\diff x-a^{T}S\right\} \\
 &  & \le t\mid X_{1},\ldots,X_{n}\Bigg]-\Phi\left(t\right)\Bigg|=0,\ascv.
\end{eqnarray*}
The remainder of the proof uses the same argument used in the proof
of Theorem 1. 
\end{proof}

\section{Proof of Theorem 4}\label{sec:Proof-of-Theorem-4}
\begin{proof}
\textcolor{black}{First we will prove for a single slice, the support
vector machine is robust by showing the normal vectors of separate
hyperplanes $\psi$ are continuous functionals of conditional probability
of $\tilde{\theta}$ given $X_{1},\ldots,X_{n}$, where $\tilde{\theta}=I_{\left[\theta\le s\right]}-I_{\left[\theta>s\right]}$
is sliced $\theta$. Using notations in \cite{li2011principal},
Let 
\[
m\left(\psi,t,X,\tilde{\theta}\right)=\psi^{T}\Sigma\psi+\lambda\left\{ 1-\tilde{\theta}\left(\psi^{T}X-t\right)\right\} ^{+},
\]
be the sample version of the Lagrangian of SVM. The population normal
vectors $\psi$ are defined as the solution of the first order condition
of the optimization problem in SVM,
\begin{equation}
0=D_{\left(\psi,t\right)}E\left\{ m\left(\psi,t,X,\tilde{\theta}\right)\right\} =\left(2\psi^{T}\Sigma,0\right)^{T}-\lambda E\left[\left(X,-1\right)^{T}\tilde{\theta}I_{\left\{ 1-\tilde{\theta}\left(\psi^{T}X-t\right)>0\right\} }\right],\label{eq:first-order-svm}
\end{equation}
where $D_{\left(\psi,t\right)}$ are partial derivatives over $\psi$
and $t$. Let the conditional probability of $\tilde{\theta}$ given
$X_{1},\ldots,X_{n}$ be $p\left(x\right)=\prob\left(\tilde{\theta}=1\mid X_{1}=x_{1},\ldots,X_{n}=x_{n}\right)$.
We need to show that $\psi\left(p\right)$ as a functional of $p\left(x\right)$
defined by \ref{eq:first-order-svm} is continuous. The main theorem
we rely on is Theorem 3.1.2 in \cite{lebedev2003functional}. Next,
we will check the three conditions in that theorem.}

\textcolor{black}{Condition (i) is trivial. For verifying Condition
(ii), we first view $p\left(X\right)$ as an element from Banach space
$\left\{ p\left(x\right):\sup_{x}\left|p\left(x\right)\right|<\infty\right\} $,
with the norm $\sup_{x}\left|\cdot\right|$. Let 
\[
g\left(X,\psi,t\right)=\begin{cases}
2p\left(X\right)-1, & \psi^{T}X-t<1,\\
p\left(X\right), & \psi^{T}X-t\ge1,\\
p\left(X\right)-1, & \psi^{T}X-t\le-1.
\end{cases}
\]
Then the second term in \ref{eq:first-order-svm}, which is the only
term containing $p\left(X\right)$ can be written as 
\begin{eqnarray*}
 &  & E\left[\left(X,-1\right)^{T}\tilde{\theta}I_{\left\{ 1-\tilde{\theta}\left(\psi^{T}X-t\right)>0\right\} }\right]\\
 & = & E_{X}\left(\left(X,-1\right)^{T}E_{\tilde{\theta}\mid X}\left[\tilde{\theta}I_{\left\{ 1-\left(\psi X-t\right)\tilde{\theta}>0\right\} }\right]\right)\\
 & = & E_{X}\left\{ \left(X,-1\right)^{T}g\left(X,\psi,t\right)\right\} \\
 & = & \int_{\left\{ \psi^{T}X-t<1\right\} }\left(x,-1\right)^{T}\left(2p\left(x\right)-1\right)\diff F\left(x\right)+\int_{\left\{ \psi^{T}X-t\ge1\right\} }\left(x,-1\right)^{T}p\left(x\right)\diff F\left(x\right)\\
 &  & +\int_{\left\{ \psi^{T}X-t\le-1\right\} }\left(x,-1\right)^{T}\left(p\left(x\right)-1\right)\diff F\left(x\right),
\end{eqnarray*}
where $F\left(x\right)$ is the CDF of marginal distribution of $X_{1},\ldots,X_{n}$.
This is a continuous linear functional map $p\left(x\right)\mapsto E\left[\left(X,-1\right)^{T}\tilde{\theta}I_{\left\{ 1-\left(\psi X+t\right)\tilde{\theta}>0\right\} }\right]$.
So Condition (ii) holds. By Theorem 5 in \cite{li2011principal},
$D_{\left(\psi,t\right)}E\left\{ m\left(\psi,t,X,\tilde{\theta}\right)\right\} $
can be further differentiated around the solutions. Hence, Condition
(iii) holds. By Theorem 3.1.2 in \cite{lebedev2003functional}, $\psi$
is continuous in $p\left(x\right)$. }

\textcolor{black}{Let $\psi_{1}=\psi\left(\prob\left(\tilde{\theta}=1\mid X_{1},\ldots,X_{n}\right)\right)$,
$\psi_{2}=\psi\left(\prob\left(\tilde{\theta}=1\mid\Gamma_{n}^{T}\left(X_{1},\ldots,X_{n}\right)\right)\right)$,
while by Theorem 2 in \cite{li2011principal}, $\mathrm{span}\left(\psi_{2}\right)\subset\mathrm{span}\left(\Gamma_{n}\right)$.
Then for any $\varepsilon>0$, there exist an $N_{1}$ and a $\delta$
such that for any $n>N_{1}$, 
\begin{eqnarray*}
 &  & \left|\prob\left\{ \tilde{\theta}=1\mid\Gamma_{n}^{T}\left(X_{1},\ldots,X_{n}\right)\right\} -\prob\left(\tilde{\theta}=1\mid X_{1},\ldots,X_{n}\right)\right|\\
 & = & \left|\prob\left\{ \theta\le s\mid\Gamma_{n}^{T}\left(X_{1},\ldots,X_{n}\right)\right\} -\prob\left(\theta\le s\mid X_{1},\ldots,X_{n}\right)\right|<\delta,
\end{eqnarray*}
and hence $\left|\psi_{1}-\psi_{2}\right|<\varepsilon$. By Theorem
6 in \cite{li2011principal}, the sample version of normal vectors
$\hat{\psi}_{1}$ is weakly consistent to the population ones. Hence
there exists an $N_{2}\ge N_{1}$, such that for any $n>N_{2},$
\begin{eqnarray*}
 &  & \prob\left(\left|\hat{\psi}_{1}-\psi_{2}\right|\ge2\varepsilon\right)\\
 & \le & \prob\left(\left\{ \left|\hat{\psi}_{1}-\psi_{1}\right|\ge\varepsilon\right\} \cup\left\{ \left|\psi_{1}-\psi_{2}\right|\ge\varepsilon\right\} \right)\\
 & = & \prob\left(\left|\hat{\psi}_{1}-\psi_{1}\right|\ge\varepsilon\right)\le\eta.
\end{eqnarray*}
Hence $\hat{\psi}_{1}-\psi_{2}\overset{p}{\rightarrow}0.$ By Theorem
1 in \cite{bura2008distribution}, $\hat{\Gamma}_{n}$ and $\Gamma_{n}$
being eigenvectors satisfy $\hat{\Gamma}_{n}-\Gamma_{n}\overset{p}{\rightarrow}0$. }
\end{proof}

\section*{Acknowledgement}
 Ghosh's research was partially supported by an NSF Grant.

\bibliography{partialPostBib}
\bibliographystyle{imsart-nameyear}


\end{document}